\documentclass[lettersize,journal,12pt,draftclsnofoot,onecolumn]{IEEEtran}
\usepackage{amsmath,amsfonts}
\usepackage{algorithmic}
\usepackage[ruled]{algorithm2e}
\usepackage{array}
\usepackage{subcaption}
\usepackage{textcomp}
\usepackage{stfloats}
\usepackage{url}
\usepackage{verbatim}
\usepackage{graphicx}
\usepackage{cite}
\usepackage[binary-units=true]{siunitx}
\usepackage{mathextensions}
\usepackage{globalglossary}
\usepackage{amsthm}
\usepackage{threeparttable}
\usepackage[draft,nomargin,footnote]{fixme}
\usepackage{pgfplots}
\usepackage{pgfstyles}
\usepgfplotslibrary{external}
\theoremstyle{remark}
\newtheorem{lemma}{Lemma}
\newtheorem{remark}{Remark}
\newtheorem{corollary}{Corollary}

\SetAlFnt{\footnotesize}
\begin{document}

\bstctlcite{IEEEexample:BSTcontrol}

\title{RIS Assisted Device Activity Detection with Statistical Channel State Information}

\author{Friedemann~Laue,~\IEEEmembership{Graduate~Student~Member,~IEEE}, Vahid~Jamali,~\IEEEmembership{Member,~IEEE}, and Robert~Schober,~\IEEEmembership{Fellow,~IEEE}
  % <-this % stops a space
\thanks{This work was presented in part at IEEE SPAWC 2021 \cite{laue2021irsassistedactive}. \textit{(Corresponding author: Friedemann Laue.)}}% <-this % stops a space
\thanks{F. Laue and R. Schober are with Institute for Digital Communications, Friedrich-Alexander-Universität Erlangen-Nürnberg (FAU), 91058 Erlangen, Germany (e-mail: friedemann.laue@fau.de, robert.schober@fau.de).}%
\thanks{F. Laue is also with Fraunhofer IIS, Fraunhofer Institute for Integrated Circuits IIS, Division Communication Systems (e-mail: friedemann.laue@iis.fraunhofer.de).}%
\thanks{V. Jamali is with Department of Electrical Engineering and Information Technology, Technical University of Darmstadt, 64283 Darmstadt, Germany (e-mail: vahid.jamali@tu-darmstadt.de).}}

\maketitle

\begin{abstract}
  This paper studies \gls{ris} assisted device activity detection for \gls{gf} uplink transmission in wireless communication networks.
  In particular, we consider mobile devices located in an area where the direct link to an \gls{ap} is blocked.
  Thus, the devices try to connect to the \gls{ap} via a reflected link provided by an \gls{ris}.
  Therefore, for the \gls{ris}, a phase-shift design is desired that covers the entire blocked area with a wide reflection beam because the exact locations and times of activity of the devices are unknown in \gls{gf} transmission.
  In order to study the impact of the phase-shift design on the device activity detection at the \gls{ap}, we derive a \gls{glrt} based detector and present an analytical expression for the probability of detection, which is a function of the channel statistics and the phase-shift design.
  Assuming knowledge of statistical \gls{csi}, we formulate an optimization problem for the phase-shift design for maximization of the guaranteed probability of detection for all locations within a given coverage area.
  To tackle the non-convexity of the problem, we propose two different approximations of the objective function and an algorithm based on the \gls{mm} principle.
  The first approximation leads to a design that aims to reduce the variations of the end-to-end channel while taking system parameters such as transmit power, noise power, and probability of false alarm into account.
  The second approximation can be adopted for versatile \gls{ris} deployments because it only depends on the \gls{los} component of the end-to-end channel and is not affected by system parameters.
  For comparison, we also consider a phase-shift design maximizing the average channel gain and a baseline analytical phase-shift design for large blocked areas.
  Our performance evaluation shows that the proposed approximations result in phase-shift designs that guarantee high probability of detection across the coverage area and outperform the baseline designs.
\end{abstract}

\begin{IEEEkeywords}
Device activity detection, reconfigurable intelligent surface, coverage extension, grant-free uplink, statistical channel state information, phase-shift design.
\end{IEEEkeywords}

\glsresetall

% \listoffixmes

\section{Introduction}
\IEEEPARstart{T}{unable} metasurfaces that control the reflection properties of electromagnetic waves were proposed several years ago \cite{kaina2014shapingcomplexmicrowave}.
Recently, this technology has received significant attention in the wireless communications community.
Its ability to create a smart radio environment helps in meeting the high performance requirements of future wireless networks \cite{renzo2019smartradioenvironments}.
To this end, metasurfaces are deployed as controllable passive reflectors, also known as \glspl{ris}.
An \gls{ris} can be modelled as an array of small tunable elements, which are called \emph{unit cells}. While an impinging electromagnetic wave is reflected on the surface of the \gls{ris}, the unit cells modify the characteristics of the wave, e.g., by applying individual phase shifts.
Thereby, the configuration of the unit cells, or the \emph{phase-shift design}, has an impact on the direction of reflection.
The potential of this technology for practical deployments was demonstrated with first \gls{ris} prototypes.
For example, indoor experiments with an \gls{ris} based on \gls{pin} diodes showed an antenna gain of up to \SI{21.7}{\deci\bel}i \cite{dai2020reconfigurableintelligentsurface}. Furthermore, an \gls{ris} employing varactor diodes and more than 1000 unit cells was successfully deployed to transmit a 1080p video stream, where a power consumption of only \SI{1}{\watt} was observed \cite{pei2021risaidedwireless}.

A major challenge in the design of \gls{ris} assisted communication systems is the efficient computation of the phase shifts that maximize the desired system performance metric.
One typical \gls{ris} use case are multi-user wireless communication networks, where an \gls{ris} is deployed to assist multiple users in performing uplink transmission to an \gls{ap} \cite{you2021energyefficiencyspectral,wei2021channelestimationris,liu2021intelligentreflectingsurface,luo2021spatialmodulationris,zhi2021uplinkachievablerate,xu2021risenhancedwpcns,you2021reconfigurableintelligentsurfaces}.
For example, an optimization framework for maximization of both the energy efficiency and the spectral efficiency of \gls{ris} assisted \gls{mimo} uplink transmission was developed in \cite{you2021energyefficiencyspectral}.
This framework was used to provide optimized phase-shift designs while exploiting the instantaneous and statistical \gls{csi} of the cascaded \gls{ris} channel.
Moreover, power minimization in an \gls{ris} assisted multi-user \gls{miso} network was considered in \cite{liu2021intelligentreflectingsurface}.
After presenting a feasibility condition for a guaranteed information rate, an algorithm was developed that accounts for different \gls{qos} constraints for perfect and imperfect \gls{csi}.
In addition, the rate performance of uplink \gls{miso} systems employing imperfect hardware was studied in \cite{zhi2021uplinkachievablerate} assuming perfect \gls{csi}.
While impairments at both the \gls{ap} and the \gls{ris} were considered, it was shown that the rate loss is mainly determined by the hardware limitations of the \gls{ap} for typical numbers of \gls{ris} unit cells.

Whereas the above works assume scheduled uplink transmission and (imperfect) knowledge of the instantaneous \gls{csi}, in this paper, we focus on \gls{gf} uplink transmission where the time of transmission is unknown at the \gls{ap}.
Consequently, device activity detection is required where only statistical \gls{csi} obtained from previous transmissions can be exploited.

\Gls{gf} transmission has been studied extensively in the literature and the main focus has recently been on device activity detection for massive connectivity \cite{liu2018massiveconnectivitymassive,ding2019sparsitylearningbased,shao2020dimensionreductionbased,zhang2021jointactiveuser,jiang2021mlmapdevice}.
For example, the authors of \cite{liu2018massiveconnectivitymassive} analyzed \glsxtrlong{cs} techniques for device activity detection exploiting the sparse activity patterns in massive connectivity.
It was shown that the probabilities of false alarm and missed detection asymptotically go to zero as the number of \gls{ap} antennas goes to infinity.
In addition, the authors of \cite{ding2019sparsitylearningbased} proposed \glsxtrlong{mud} based on sparse user activity without the need for pilot signals in the context of massive connectivity.
Furthermore, the authors of \cite{shao2020dimensionreductionbased} developed an optimization framework for joint device activity detection and channel estimation, where the high complexity caused by large-scale antenna arrays and huge numbers of devices was made tractable by dimension reduction of the feasible space.
Similarly, joint device activity detection and channel estimation was analyzed in \cite{zhang2021jointactiveuser}.
More specifically, two \glsxtrlong{cs} based methods for joint multi-user detection and channel estimation in \gls{gf} \glsxtrlong{noma} systems were presented.
Furthermore, the analysis of \gls{gf} massive access was extended to multi-cell networks in \cite{jiang2021mlmapdevice}. In particular, device activities and interference powers were estimated using both cooperative and noncooperative approaches based on \glsentrylong{ml} and \glsentrylong{map} estimation.

Moreover, there is a small number of works that have investigated \gls{gf} transmission and device activity detection in \gls{ris} assisted networks.
For example, three recent studies focused on \gls{gf} access for \gls{ris} assisted \glsentrylong{iot} networks \cite{guo2021sparseactivitydetection,xia2021reconfigurableintelligentsurface,shao2021bayesiantensorapproach}.
More specifically, assuming \gls{iid} channel coefficients, sparse device activity, and large \glspl{ris}, the authors of \cite{guo2021sparseactivitydetection} showed that the end-to-end channel can be characterized by a complex Gaussian distribution that is independent of the phase-shift design.
This result was adopted to study \glsxtrlong{ce} and device activity detection using a \glsxtrlong{gamp} algorithm.
Similarly, the goal in \cite{xia2021reconfigurableintelligentsurface} was to detect a set of active devices in a massive connectivity scenario and to estimate the \gls{ris} assisted channels.
By exploiting the sparsity of both the channel and the device activity, an \glsxtrlong{amp} based algorithm for joint activity detection and \glsxtrlong{ce} was developed assuming identical phase shifts for all unit cells.
Furthermore, the authors of \cite{shao2021bayesiantensorapproach} focused on unsourced random access, which is an effective random access scheme for massive connectivity. In this context, joint user separation and channel estimation was supported by an \gls{ris} with a random phase-shift configuration.

As is evident from the above discussion, the phase shifts of the \glspl{ris} in \cite{guo2021sparseactivitydetection,xia2021reconfigurableintelligentsurface,shao2021bayesiantensorapproach} were not specifically designed for the problem of device activity detection.
However, the phase-shift design can have a significant impact on the detection performance in practical systems where signals propagate on a limited number of paths.
In this case, it is important to apply a phase-shift design that takes the incident and reflected angles at the \gls{ris} into account \cite{najafi2020physicsbasedmodeling}.
Otherwise, device activity detection at the \gls{ap} may fail because the devices' locations are not covered by the \gls{ris} reflection pattern \cite{popov2021experimentaldemonstrationmmwave}.

To this end, the authors of the recent work \cite{croisfelt2022randomaccessprotocol} proposed a random access protocol that sweeps through a finite set of \gls{ris} configurations. In particular, a downlink phase and an uplink phase are used to estimate the channel and to connect to the \gls{ap}, respectively.
Although this protocol significantly improves the access performance, the model in \cite{croisfelt2022randomaccessprotocol} assumes a simplified setup with \gls{los} links and one-dimensional phase-shift configurations.
In addition, the protocol causes access delays due to the training phase, and device activity detection at the \gls{ap} was not considered in \cite{croisfelt2022randomaccessprotocol}.

In this paper, we study \gls{ris} phase-shift design for \gls{gf} access and device activity detection, and extend the analysis of our previous work \cite{laue2021irsassistedactive}.
In particular, we consider a communication system where the direct link from the \gls{ap} to a pre-defined area is blocked.
Thus, an \gls{ris} is deployed to connect the \gls{ap} with mobile devices located in the blocked area through a reflected link.
In contrast to \cite{guo2021sparseactivitydetection,xia2021reconfigurableintelligentsurface,shao2021bayesiantensorapproach,croisfelt2022randomaccessprotocol}, in this paper, we optimize the phase-shift design to maximize the minimum probability of detection for any device within the blocked area.
Intuitively, since the devices' exact locations and times of transmission are unknown for \gls{gf} access, the challenge is to develop a phase-shift design that covers the entire area with a wide reflection beam.

To this end, we propose a \gls{glrt} based detector and formulate an optimization problem that maximizes the probability of detection across all locations within the blocked area.
Since the optimization problem is non-convex and involves the Marcum Q function,
we reformulate the problem considering two approximations of the Marcum Q function and solve it applying the \gls{mm} principle, where each approximation results in a different phase-shift design.
Finally, we adopt the proposed \gls{glrt} detector to evaluate the detection performance of these designs.
For comparison, we consider an alternative design approach that maximizes the average channel gain and a baseline analytical design.

The main contributions of this paper can be summarized as follows:
\begin{itemize}
  \item We investigate device activity detection for \gls{ris} assisted communication systems using a geometric channel model based on statistical \gls{csi}.
  Our analysis includes both the \gls{los} and \gls{nlos} paths of the device-\gls{ris} channels and takes into account their incident and reflected angles at the \gls{ris}. To this end, we adopt a physics-based \gls{ris} model for characterization of the angle-dependent reflection gain.
  \item Based on the \gls{glrt}, we derive a device activity detector and analyze its performance.
  In particular, we show that the probability of detection can be expressed in terms of the Marcum Q function and it depends on the phase-shift design of the \gls{ris} and the statistics of the channel.
  \item We formulate the phase-shift design as an optimization problem for maximization of the minimum probability of detection across all locations within a given coverage area.
  In order to tackle the non-convexity of the problem, we propose two approximations for the objective function and reformulate them as \glspl{dc} functions.
  This allows us to solve the problem exploiting the \gls{mm} principle and obtain a high-quality suboptimal phase-shift design for each approximation.
  \item We evaluate the detection performance for the proposed design approach and show that our method outperforms existing schemes for \gls{ris} phase-shift design.
  In particular, the proposed design approach improves the probability of device activity detection compared to a design maximizing the average channel gain \cite{shen2021beamformingoptimizationirs,xing2021achievablerateanalysis,xing2021locationawarebeamforming} and an analytical phase-shift design for large coverage areas \cite{laue2021irsassistedactive,jamali2021powerefficiencyoverhead}.
\end{itemize}

Different from its conference version \cite{laue2021irsassistedactive}, this paper takes both the \gls{los} and \gls{nlos} paths of the device-\gls{ris} channels into account.
As a result, the detection performance cannot be expressed as a monotonic function of the \gls{los} channel gain as in \cite{laue2021irsassistedactive}, i.e., a different approach for phase-shift design is required.
Similar to \cite{laue2021irsassistedactive}, the design proposed in this work is based on a non-convex optimization problem. However, instead of resolving the non-convexity of the problem using semidefinite relaxation and Gaussian randomization, we rewrite the non-convex terms in \gls{dc} form and apply the \gls{mm} principle.
In addition, we investigate the impact of different transmit powers, scattering strengths, incident angles at the \gls{ris}, and area sizes. We also study the reflection patterns of the optimized phase-shift designs.

The remainder of this paper is organized as follows.
Section \ref{sec:system-model} introduces the system and channel models.
The \gls{glrt} for the considered device activity detection problem is derived in Section \ref{sec:active-device-detection}. We formulate the optimization problem for the phase-shift design in Section \ref{sec:phase-shift-design}, and present the \gls{mm} based solution in Section \ref{sec:successive-convex-approximation}.
The performance of the proposed phase-shift designs is evaluated in Section \ref{sec:performance-evaluation}, and conclusions are drawn in Section \ref{sec:conclusion}.

\emph{Notations}: $\tr{\cdot}$, $(\cdot)^T$, $(\cdot)^H$, $(\cdot)^{-1}$, and $\E{\cdot}$ denote the trace, transpose, conjugate transpose, inverse, and expectation operations, respectively.
$\norm{\cdot}$ refers to the Euclidean norm of a vector and $\abs*{\cdot}$ represents the absolute value of a scalar.
The nuclear norm and spectral norm of a matrix are denoted by $\norm{\cdot}_*$ and $\norm{\cdot}_2$, respectively.
Given a function $f(\mt{X})$ with matrix argument $\mt{X}$, $\nabla_{\mt{X}}f$ represents the gradient of $f(\mt{X})$.
The element in the $m$th row and $n$th column of a matrix and the $n$th element of a vector are denoted by $\left[\,\cdot\,\right]_{m\times n}$ and $\left[\,\cdot\,\right]_{n}$, respectively.
The determinant of matrix $\mt{X}$ is denoted by $\det(\mt{X})$, the identity matrix is denoted by $\mt{I}$, and a positive semidefinite matrix $\mt{X}$ is denoted by $\mt{X}\succeq 0$.
$\mathbb{R}$ and $\mathbb{C}$ represent the sets of real and complex numbers, respectively.
$\vt{x}\sim\mathcal{CN}(\vt{a}, \mt{B})$ refers to a complex normally distributed random vector $\vt{x}$ with mean vector $\vt{a}$ and covariance matrix $\mt{B}$.
$\chi_a^2(b)$ represents a noncentral chi-squared distribution with $a$ degrees of freedom and noncentrality parameter $b$.
Probability is denoted by $\Pr\{\cdot\}$.
The operator $\floor{x}$ represents the largest integer less than or equal to $x$ and the remainder of the division $a/b$ is denoted by $a\pmod{b}$.
The exponential function $e^{(\cdot)}$ is also written as $\exp(\cdot)$.

\section{System Model}\label{sec:system-model}
\noindent In this section, we present the considered communication network, where the direct link from the \gls{ap} to a specific area is blocked.
As illustrated in Fig.~\ref{fig:illustration}, mobile devices located within the blocked area connect to the \gls{ap} via the reflected link of the \gls{ris}.
In the following, we introduce the coordinate system for the communication network and describe the coverage area, the \gls{ris}, and the \gls{ap} in more detail. Additionally, we present the adopted channel model.
\begin{figure}
  \centerline{\includegraphics[width=\columnwidth]{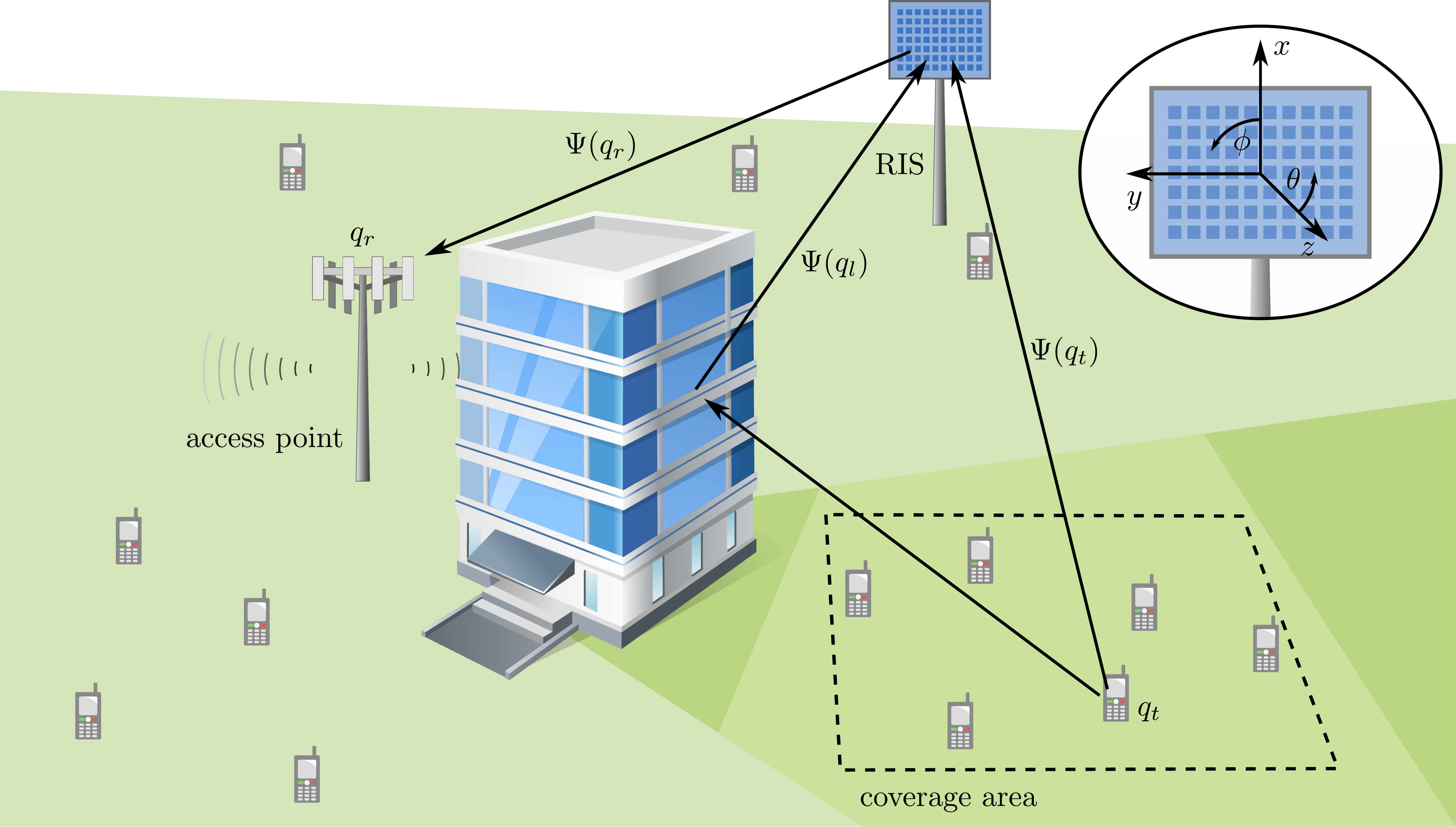}}
  \caption{Illustration of the considered communication setup for device activity detection.}
  \label{fig:illustration}
\end{figure}

\subsection{Coordinate System}
\noindent We place the origin of the coordinate system at the center of the \gls{ris} and refer to a location in three-dimensional space as $q$.
Location $q$ is specified by Cartesian coordinates $\mathsf{cart}(q)=(x_q, y_q, z_q)$ or spherical coordinates $\mathsf{sphe}(q) = (d_q, \theta_q, \phi_q)$, where $d_q$, $\theta_q$, and $\phi_q$ denote the distance from the origin to $q$, the elevation angle of $q$, and the azimuth angle of $q$, respectively.
Further, we denote the direction from the origin to location $q$ as $\Psi(q) = (\theta_q, \phi_q)$.

\subsection{Coverage Area and Transmitters}
\noindent The actual coverage area is an arbitrary area in continuous three-dimensional space, e.g., an area in the $y$-$z$ plane. For tractability, we take $Q$ samples of the continuous area and define the approximated coverage area as the discrete set of locations $\mathcal{Q}=\left\{q_{t,1}, q_{t,2}, \dots,q_{t,Q}\right\}$. The number of samples can be chosen according to the desired approximation accuracy \cite{laue2021irsassistedactive}.
For simplicity, we assume a rectangular coverage area with center coordinates $(c_x, c_y, c_z)$, length $D_z$ in $z$-direction, and width $D_y$ in $y$-direction.

The location of a device is denoted as $q_t$, where $q_t\in\mathcal{Q}$ holds for devices located within the coverage area.
We consider single-antenna devices that try to connect to the \gls{ap} by transmitting a preamble sequence, which is defined by symbol vector $\vt{s}\in\mathbb{C}^S$ of length $S$. We assume normalized symbols, i.e., $\abs*{\left[\vt{s}\right]_s}^2 = 1, \forall\ s\in\{1,\dots,S\}$.

\subsection{Access Point}\label{sec:accesspoint}
\noindent We consider an $M$-antenna \gls{ap} deployed at a fixed and known location $q_r$.
Furthermore, we assume a linear receiver with coherent combining and a \gls{los} link\footnote{
In \gls{ris} assisted networks, it is usually desired that the \gls{ris} is deployed with dominant \gls{los} to the \gls{ap} in order to achieve the maximum received power \cite{wu2018intelligentreflectingsurface,wu2020towardssmartreconfigurable}. Thus, similar to \cite{wu2018intelligentreflectingsurface,wu2020towardssmartreconfigurable,shen2021beamformingoptimizationirs,xing2021achievablerateanalysis,croisfelt2022randomaccessprotocol}, we assume that in the \gls{ap}-\gls{ris} channel the impact of fading is negligible and consider only the \gls{los} path.}to the \gls{ris}.
Thus, the \gls{ap} can be modelled as an equivalent single-antenna receiver with beamforming gain $M$.
The transmitted signal of an active device is received via the \gls{ris} at the \gls{ap} as\footnote{For the phase-shift design of the \gls{ris}, we assume that only one device is active at a time.
However, the phase-shift designs proposed in this paper provide wide reflection beams that cover a given area, i.e., these designs extend the coverage of the \gls{ap}.
Thus, the designs can be straightforwardly combined with existing schemes for \glsxtrlong{mud}, \glsxtrlong{ce}, and data transmission designed for non-\gls{ris} assisted multi-user networks.}
\begin{equation}\label{eq:rx-signal}
  \vt{x} = h\sqrt{MP_\text{tx}}\vt{s} + \vt{n},
\end{equation}
where $h\in\mathbb{C}$, $P_\text{tx}\in\mathbb{R}$, and $\vt{n}\in\mathbb{C}^{S}$ denote the end-to-end channel gain, the transmit power, and additive Gaussian noise, respectively.
The noise is modelled as $\vt{n}\sim\CN{\vt{0}, \sigma^2\mt{I}}$ with noise power $\sigma^2$.
Note that $h$ depends on the phase-shift design and the channel model, which is presented in Section \ref{sec:channel-model}.

\subsection{Reconfigurable Intelligent Surface}
\noindent It has been shown that the reflection gain of the \gls{ris} depends on the \glspl{aoa} and \glspl{aod} of the incident and reflected waves, respectively \cite{najafi2020physicsbasedmodeling,pei2021risaidedwireless,wang2021receivedpowermodel}.
Therefore, in order to correctly model the reflection gain for the different device locations, we adopt the physics-based \gls{ris} model from \cite{najafi2020physicsbasedmodeling}.
In particular, we assume that the \gls{ris} is centered at the origin of the coordinate system and is modelled as a \gls{upa} with $U_x U_y = U$ unit cells in the $x$-$y$ plane. The location of unit cell $u \in\{0, 1, \dots, U-1\}$ is given by the coordinate vector
\begin{equation}
  \vt{c}_u = \begin{bmatrix}d_x u_x & d_y u_y & 0\end{bmatrix}^T,
\end{equation}
where $d_x$ and $d_y$ denote the unit cell spacing along the $x$ and $y$ axis, respectively, and
\begin{align}
  u_x &= u\pmod{U_x} - U_x/2 + 1\\
  u_y &= \floor{u/U_x} - U_y/2 + 1.
\end{align}
Here, we assume that $U_x$ and $U_y$ are even numbers.
The response function of the \gls{ris} is given by \cite{najafi2020physicsbasedmodeling}
\begin{equation}\label{eq:irs-response}
  g(\Psi(q_r), \Psi(q_t)) = \frac{\sqrt{4\pi}}{\lambda} c(\Psi(q_r), \Psi(q_t)) \vt{w}^H \vt{a}(\Psi(q_r), \Psi(q_t)),
\end{equation}
where $\lambda\in\mathbb{R}, c(\Psi(q_r), \Psi(q_t))\in\mathbb{C}$, $\vt{w}\in\mathbb{C}^U$, and $\vt{a}(\Psi(q_r), \Psi(q_t))\in\mathbb{C}^U$ denote the wavelength, the unit cell factor, the phase-shift vector, and the array response, respectively.
Defining
\begin{equation}
  \vt{k}(\Psi(q)) = \frac{2\pi}{\lambda}
    \begin{bmatrix}
      \sin(\theta_q)\cos(\phi_q) & \sin(\theta_q)\sin(\phi_q) & \cos(\theta_q)
    \end{bmatrix}^T
\end{equation}
and $\vt{k}(q_r, q_t) = \vt{k}(\Psi(q_r)) + \vt{k}(\Psi(q_t))$, the array response can be expressed as
\begin{equation}
  \vt{a}(\Psi(q_r), \Psi(q_t)) =
  \begin{bmatrix}
    e^{j\vt{k}^T(q_r, q_t) \vt{c}_0}
    & e^{j\vt{k}^T(q_r, q_t) \vt{c}_1}
    & \dots
    & e^{j\vt{k}^T(q_r, q_t) \vt{c}_{U-1}}\end{bmatrix}^T.
\end{equation}
The $u$th element of $\vt{w}$ is given by $e^{j\omega_u}$, where $\omega_u$ denotes the phase shift applied by the $u$th unit cell.
Unit cell factor $c(\Psi(q_r), \Psi(q_t))$ describes the amplitude of the reflection coefficient of each unit cell, which depends on the incident and reflected angles, the polarization of the incident wave, and the physical realization of the unit cells.
For example, using the physics-based model of \cite{najafi2020physicsbasedmodeling}, the unit cell factor is given by
\begin{multline}\label{eq:unitcellfactor}
  c(\Psi(q_r), \Psi(q_t)) = \frac{j\sqrt{4\pi}d_xd_y}{\lambda}\\
    \times
    \cos(\theta_{q_t})\frac{
      \norm*{
        \begin{bmatrix}
          \cos(\varphi_{q_t})\cos(\theta_{q_r})\sin(\phi_{q_r})
          - \sin(\varphi_{q_t})\cos(\theta_{q_r})\cos(\phi_{q_r})\\
          \sin(\varphi_{q_t})\sin(\phi_{q_r}) + \cos(\varphi_{q_t})\cos(\phi_{q_r})
        \end{bmatrix}
      }_2
    }{
      \sqrt{\left(
        \cos(\varphi_{q_t})\sin(\theta_{q_t})\cos(\phi_{q_t})
        + \sin(\varphi_{q_t})\sin(\theta_{q_t})\sin(\phi_{q_t})
      \right)^2 + \cos^2(\theta_{q_t})}
    },
\end{multline}
where $\varphi_{q_t}$ denotes the polarization of the incident wave originating from location $q_t$.

\subsection{Channel Model}\label{sec:channel-model}
\noindent We assume that the direct link between the coverage area and the \gls{ap} is blocked, cf. Fig.~\ref{fig:illustration}.
Therefore, the devices have to connect with the \gls{ap} via the reflected link, which comprises the device-\gls{ris} and the \gls{ris}-\gls{ap} links.
As stated in Section \ref{sec:accesspoint}, the \gls{ris}-\gls{ap} link is dominated by the \gls{los} path.
For the device-\gls{ris} link, the scattered paths cannot be neglected because the devices are located close to the ground \cite{rao2014distributedcompressivecsit,zhi2021reconfigurableintelligentsurface}.
Nevertheless, we assume a limited number of clusters of scatterers \cite{najafi2020physicsbasedmodeling}.
Furthermore, the clusters' incident directions at the \gls{ris} are fixed and equal for all locations within the coverage area \cite{rao2014distributedcompressivecsit}.
Therefore, the channel in \eqref{eq:rx-signal} can be modelled as follows
\begin{equation}\label{eq:channel}
  h = \bar{h}(q_r)
    \left(
      \bar{h}(q_t) g(\Psi(q_r), \Psi(q_t))
      + \sum_{l=1}^L \eta_l(q_t) g(\Psi(q_r), \Psi(q_l))
    \right),
\end{equation}
where $L$, $\Psi(q_l)$, and $\eta_l(q_t)\sim\CN{0, \sigma_l^2(q_t)}$ denote the number of clusters, the incident direction of the $l$th cluster on the \gls{ris}, and the small-scale fading with variance $\sigma_l^2(q_t)$ caused by the $l$th cluster, respectively, and $\bar{h}(q)$ denotes the \gls{los} channel coefficient between the \gls{ris} and location $q$.
Note that the variances of the small-scale fading may change for different device locations, i.e., variance $\sigma_l^2(q_t)$ depends on $q_t$.
Also, similar to \cite{yu2020designanalysisoptimization,xu2020resourceallocationirs,xia2021reconfigurableintelligentsurface,shao2021bayesiantensorapproach}, we assume narrow-band communication, i.e., the delays between different propagation paths are not resolvable.
The channel gain $\bar{h}(q)$ can be written as $\bar{h}(q)=\abs*{\bar{h}(q)} e^{j\varphi(q)}$, where magnitude $\abs*{\bar{h}(q)} = \lambda/(4\pi d_q)$ is given by the free-space path loss and depends on wavelength $\lambda$ and distance $d_q$ between location $q$ and the \gls{ris}.
The phase shift $\varphi(q)=2\pi d_q/\lambda$ depends on the wavelength and the distance as well.
\begin{remark}\label{rm:known-statistics}
  In the following, we assume that $\Psi(q_l)$ and $\sigma_l^2(q_t), \forall\ l\in\{1,\dots,L\}$, are known based on channel measurements for previous transmissions \cite{wang2019superresolutionchannel}.
\end{remark}
\begin{remark}\label{rm:unknown-phases}
  We note that, in practice, the accuracy of the (measured) distance $d_q$ is limited and may not be in the order of the wavelength. As a result, we assume that the phase terms $e^{j\varphi(q_r)}$ and $e^{j\varphi(q_t)}$ are unknown \cite{laue2021irsassistedactive}.
  In contrast, small deviations of $d_q$ do not have a significant impact on the channel magnitude. Thus, we assume that $\abs*{\bar{h}(q)}$ is known.
\end{remark}
Substituting \eqref{eq:irs-response} in \eqref{eq:channel} allows us to rewrite the channel coefficient as follows
\begin{equation}\label{eq:channel-as-inner}
    h = \vt{w}^H \left( \vt{h}^{\text{LoS}} + \vt{h}^{\text{NLoS}} \right),
\end{equation}
where
\begin{subequations}\label{eq:path-components}
  \begin{align}
    \vt{h}^{\text{LoS}}&= \frac{\sqrt{4\pi}}{\lambda} \bar{h}(q_r) \bar{h}(q_t)
      c(\Psi(q_r), \Psi(q_t)) \vt{a}(\Psi(q_r), \Psi(q_t))\label{eq:los-component}\\
    \vt{h}^{\text{NLoS}}&= \frac{\sqrt{4\pi}}{\lambda} \bar{h}(q_r)
      \sum_{l=1}^L \eta_l c(\Psi(q_r), \Psi(q_l)) \vt{a}(\Psi(q_r), \Psi(q_l)).\label{eq:nlos-component}
  \end{align}
\end{subequations}
Based on \eqref{eq:channel-as-inner}, we note that phase-shift vector $\vt{w}$ has an impact on the end-to-end channel $h$. Moreover, we conclude that $h$ can be modelled as a complex Gaussian random variable with $h\sim\CN{\vt{w}^H\vt{h}^{\text{LoS}}, \vt{w}^H\mt{C}\vt{w}}$, where
\begin{equation}\label{eq:covariance}
  \mt{C} = \frac{4\pi}{\lambda^2} \abs*{\bar{h}(q_r)}^2
    \sum_{l=1}^L \sigma_l^2(q_t) \abs*{c(\Psi(q_r), \Psi(q_l))}^2 \vt{a}(\Psi(q_r), \Psi(q_l)) \vt{a}^H(\Psi(q_r), \Psi(q_l)).
\end{equation}

\section{Device Activity Detection}\label{sec:active-device-detection}
\noindent In this section, we consider the problem of device activity detection and derive analytical expressions for the detection performance.
In general, for a given desired probability of false alarm, the maximum probability of detection of an active device is obtained with the \gls{lrt} \cite{kay1998fundamentalsstatisticalsignal}.
However, the \gls{lrt} requires full knowledge of the distribution of the end-to-end channel.
We observe from \eqref{eq:path-components} and \eqref{eq:covariance} that the variance of the channel is known, but the mean value depends on the unknown phase terms of the \gls{los} paths, cf. Remark~\ref{rm:unknown-phases}.
Therefore, we use estimates of the unknown terms for the derivation of the detector, which results in a \gls{glrt} \cite{kay1998fundamentalsstatisticalsignal}.

For derivation of the proposed detector, we exploit the following lemma.
\begin{lemma}\label{lm:mle}
  For $\vt{x},\vt{y} \in \mathbb{C}^N$, $z, \gamma \in \mathbb{C}$, $c_1,c_2 \in \mathbb{R}\setminus\{0\}$, and $\mt{C}=c_1\vt{y}\vt{y}^H + c_2\mt{I}$, the following identity holds:
  \begin{equation}\label{eq:lemma-mle}
      \min_\gamma\ (\vt{x}-\vt{y}z\gamma)^H \mt{C}^{-1} (\vt{x}-\vt{y}z\gamma)
      = (\vt{x}-\frac{1}{\vt{y}^H\vt{y}}\vt{y}\vt{y}^H\vt{x})^H \mt{C}^{-1} (\vt{x}-\frac{1}{\vt{y}^H\vt{y}}\vt{y}\vt{y}^H\vt{x}).
  \end{equation}
\end{lemma}

\begin{proof}
  Given $\mt{H}\in\mathbb{C}^{N\times p}$ and $\mt{C}\in\mathbb{C}^{N\times N}$, let $\vt{x}\sim\CN{\mt{H}\vt{\theta}, \mt{C}}$ with the unknown vector $\vt{\theta}\in\mathbb{C}^p$.
  Then, the \gls{mle} of $\vt{\theta}$ is given by \cite[Chapter 15]{kay1993fundamentalsstatisticalprocessing}
  \begin{equation}
    \hat{\vt{\theta}} = \left(\mt{H}^H\mt{C}^{-1}\mt{H}\right)^{-1}\mt{H}^H\mt{C}^{-1}\vt{x}.
  \end{equation}
  The left-hand side of \eqref{eq:lemma-mle} is equivalent to finding the \gls{mle} of $\gamma$ for $\vt{x}\sim\CN{\vt{y}z\gamma, \mt{C}}$. Hence, with $\mt{H}=\vt{y}z$ and $\vt{\theta}=\gamma$, we obtain the \gls{mle}
  \begin{equation}\label{eq:mle}
    \hat{\gamma} = \frac{1}{z}\frac{\vt{y}^H\mt{C}^{-1}\vt{x}}{\vt{y}^H\mt{C}^{-1}\vt{y}}.
  \end{equation}
  The inverse of $\mt{C}=c_1\vt{y}\vt{y}^H + c_2\mt{I}$ is found with the Woodbury matrix identity (matrix inversion lemma) as
  \begin{equation}\label{eq:invcov}
    \mt{C}^{-1} = \frac{1}{c_2}\left(\mt{I} - \frac{c_1\vt{y}\vt{y}^H}{c_2 + c_1\vt{y}^H\vt{y}}\right).
  \end{equation}
  Substituting \eqref{eq:invcov} in \eqref{eq:mle} results in
  \begin{equation}\label{eq:gamma-mle}
    \hat{\gamma} = \frac{1}{z}\frac{\vt{y}^H\vt{x}}{\vt{y}^H\vt{y}},
  \end{equation}
  which gives the value for $\gamma$ that minimizes the left-hand side of \eqref{eq:lemma-mle}. The right-hand side is obtained by substituting \eqref{eq:gamma-mle} for $\gamma$ in \eqref{eq:lemma-mle}.
\end{proof}

\subsection{Detection Problem}
\noindent Device activity detection at the \gls{ap} can be described as a binary hypothesis test problem\cite{kay1998fundamentalsstatisticalsignal}.
More specifically, we define hypothesis $\mathcal{H}_0$ for an inactive device and hypothesis $\mathcal{H}_1$ for an active device.
Hence, received signal $\vt{x}$ can be written as
\begin{subequations}\label{eq:hypotheses}
  \begin{align}
    \mathcal{H}_0&: \vt{x} = \vt{n}\label{eq:h0}\\
    \mathcal{H}_1&: \vt{x} = h\sqrt{P}\bar{\vt{s}} + \vt{n},\label{eq:h1}
  \end{align}
\end{subequations}
where \eqref{eq:h0} contains noise only and \eqref{eq:h1} follows from \eqref{eq:rx-signal} with $P=MP_\text{tx}S$ and $\bar{\vt{s}}=\sqrt{\frac{1}{S}}\vt{s}$.
Here, we normalize $\vt{s}$ to obtain $\norm{\bar{\vt{s}}}^2=1$, which simplifies the derivations in the following sections.
Using the channel model in \eqref{eq:channel-as-inner}, we obtain
\begin{subequations}\label{eq:hypotheses-distribution}
  \begin{align}
    \mathcal{H}_0&: \vt{x} \sim \CN{0, \mt{C}_0}\\
    \mathcal{H}_1&: \vt{x} \sim \CN{\vt{\mu}(\gamma), \mt{C}_1},
  \end{align}
\end{subequations}
where
\begin{align}
  \gamma &=e^{j(\varphi(q_r) + \varphi(q_t))}\label{eq:unknowns}\\
  \vt{\mu}(\gamma)&= \vt{w}^H \frac{\sqrt{4\pi}}{\lambda} \abs*{\bar{h}(q_r)}\abs*{\bar{h}(q_t)} \gamma
    c(\Psi(q_r), \Psi(q_t)) \vt{a}(\Psi(q_r), \Psi(q_t))
    \sqrt{P}\bar{\vt{s}}\label{eq:mean}\\
  \mt{C}_0&= \sigma^2\mt{I}\\
  \mt{C}_1&= \vt{w}^H\mt{C}\vt{w}P\bar{\vt{s}}\bar{\vt{s}}^H + \sigma^2\mt{I}\label{eq:covh1}.
\end{align}
In \eqref{eq:unknowns} and \eqref{eq:mean}, we introduce parameter $\gamma$ comprising the phase terms of the \gls{los} links, which constitutes the unknown parameter for the \gls{glrt}.

\subsection{Detection Metric}
\noindent In general, the \gls{glrt} of a binary hypothesis test problem is based on a detection metric $T(\vt{x})$ that is derived from the likelihood ratio of both hypotheses.
This metric is compared to a threshold $t$ and the detector assumes that $\mathcal{H}_1$ is the true hypothesis when $T(\vt{x}) > t$ holds for received signal $\vt{x}$ \cite{kay1998fundamentalsstatisticalsignal}.

For the problem at hand, the log-likelihood ratio for the hypotheses in \eqref{eq:hypotheses} is given by \cite{kay1998fundamentalsstatisticalsignal}
\begin{equation}\label{eq:log-likelihood-ratio}
  L_G(\vt{x}, \gamma) = \log\left(
      \frac{\det(\mt{C}_0)}{\det(\mt{C}_1)}
      \frac{\exp\left(
        -(\vt{x}-\vt{\mu}(\gamma))^H
        \mt{C}_1^{-1}
        (\vt{x}-\vt{\mu}(\gamma))
        \right)}
        {\exp\left(-\vt{x}^H\mt{C}_0^{-1}\vt{x}\right)}
      \right).
\end{equation}
The dependence on unknown parameter $\gamma$ is resolved by finding the best estimate $\hat{\gamma}$ for detection, i.e., finding $\hat{\gamma}$ that maximizes \eqref{eq:log-likelihood-ratio} \cite{kay1998fundamentalsstatisticalsignal}.
Using Lemma \ref{lm:mle}, we obtain
\begin{subequations}\label{eq:llr-simple}
  \begin{align}
    \max_\gamma L_G(\vt{x}, \gamma)&= \log\left(
      \frac{\det(\mt{C}_0)}{\det(\mt{C}_1)}
      \frac{\exp\left(
        -(\vt{x}-\bar{\vt{s}}\bar{\vt{s}}^H\vt{x})^H
        \mt{C}_1^{-1}
        (\vt{x}-\bar{\vt{s}}\bar{\vt{s}}^H\vt{x})
        \right)}
        {\exp\left(-\vt{x}^H\mt{C}_0^{-1}\vt{x}\right)}
      \right)\\
      &= \log\left(\frac{\det(\mt{C}_0)}{\det(\mt{C}_1)}\right)
        + \vt{x}^H \mt{C}_0^{-1} \vt{x} - (\vt{x}-\bar{\vt{s}}\bar{\vt{s}}^H\vt{x})^H \mt{C}_1^{-1}
        (\vt{x}-\bar{\vt{s}}\bar{\vt{s}}^H\vt{x}).\label{eq:llr-simpler}
    \end{align}
\end{subequations}
The first term in \eqref{eq:llr-simpler} is independent of $\vt{x}$ and can be neglected for detection. Thus, the detection metric of the \gls{glrt} is given by
\begin{subequations}\label{eq:metric-complex}
  \begin{align}
    T(\vt{x})&= \vt{x}^H \mt{C}_0^{-1} \vt{x} - (\vt{x}-\bar{\vt{s}}\bar{\vt{s}}^H\vt{x})^H
    \mt{C}_1^{-1}
    (\vt{x}-\bar{\vt{s}}\bar{\vt{s}}^H\vt{x})\\
    &= \vt{x}^H (\mt{C}_0^{-1} - \mt{C}_1^{-1}) \vt{x}
      - \vt{x}^H\bar{\vt{s}}\bar{\vt{s}}^H \mt{C}_1^{-1} \bar{\vt{s}}\bar{\vt{s}}^H\vt{x}
      + \vt{x}^H \mt{C}_1^{-1} \bar{\vt{s}}\bar{\vt{s}}^H\vt{x}
      + \vt{x}^H\bar{\vt{s}}\bar{\vt{s}}^H \mt{C}_1^{-1} \vt{x}.
  \end{align}
\end{subequations}
Since $\mt{C}_0^{-1} = \sigma^{-2}\mt{I}$ and $\mt{C}_1^{-1}$ is found with \eqref{eq:invcov} and \eqref{eq:covh1}, \eqref{eq:metric-complex} simplifies to
\begin{equation}\label{eq:metric}
  T(\vt{x}) = \frac{1}{\sigma^2} \abs*{\bar{\vt{s}}^H\vt{x}}^2,
\end{equation}
which is a correlation detector for the known preamble sequence $\bar{\vt{s}}$.

\subsection{Detection Performance}
\noindent In order to determine the performance of detection metric \eqref{eq:metric}, we examine the distribution of $T(\vt{x})$.
Based on \eqref{eq:hypotheses} and \eqref{eq:hypotheses-distribution}, \eqref{eq:metric} can be written as
\begin{subequations}\label{eq:metric-hypotheses}
  \begin{align}
    \mathcal{H}_0: T(\vt{x})&=
      \frac{1}{\sigma^2} \abs*{\bar{\vt{s}}^H\vt{n}}^2 =
      \frac{1}{2}z_0\\
    \mathcal{H}_1: T(\vt{x})&=
      \frac{1}{\sigma^2}
        \abs*{\bar{\vt{s}}^H\left(h\sqrt{P}\bar{\vt{s}} + \vt{n}\right)}^2 =
      \frac{\vt{w}^H\mt{C}\vt{w}P + \sigma^2}{2\sigma^2}z_1,
  \end{align}
\end{subequations}
where
\begin{subequations}\label{eq:metric-distributions}
  \begin{align}
    z_0 &= \frac{2}{\sigma^2} \abs*{\bar{\vt{s}}^H\vt{n}}^2 \sim \chi_2^2(0)\\
    z_1 &= \frac{2}{\vt{w}^H\mt{C}\vt{w}P + \sigma^2} \abs*{\bar{\vt{s}}^H\left(h\sqrt{P}\bar{\vt{s}} + \vt{n}\right)}^2 \sim \chi_2^2\left(2P\frac{\abs*{\vt{w}^H\vt{h}^{\text{LoS}}}^2}{\vt{w}^H\mt{C}\vt{w}P + \sigma^2}\right).
  \end{align}
\end{subequations}
Thus, in general, the detection metric follows a scaled noncentral chi-squared distribution, where both the scaling factor and the noncentrality parameter take different values under $\mathcal{H}_0$ and $\mathcal{H}_1$.
The noncentrality parameter under $\mathcal{H}_0$ equals zero, i.e., the probability of false alarm is given by
\begin{equation}\label{eq:pf}
    P_F= \text{Pr}\left\{\frac{1}{2} z_0 > t\right\} = 1 - \left(1 - e^{-t}\right).
\end{equation}
Thus, for a desired $P_F$, the detection threshold is obtained as $t = -\ln(P_F)$.
The probability of detection is then given by
\begin{equation}\label{eq:pd}
    P_D = \text{Pr}\left\{\frac{\vt{w}^H\mt{C}\vt{w}P + \sigma^2}{2\sigma^2} z_1 > t\right\}
    = \text{Pr}\left\{z_1 > \frac{- 2\sigma^2 \ln(P_F)}{\vt{w}^H\mt{C}\vt{w}P + \sigma^2}\right\}.
\end{equation}
Since $z_1$ is a noncentral chi-squared random variable, \eqref{eq:pd} can be expressed as follows
\begin{equation}\label{eq:probability-of-detection}
  P_D = Q_1\left(\sqrt{2P\frac{\abs*{\vt{w}^H\vt{h}^{\text{LoS}}}^2}{\vt{w}^H\mt{C}\vt{w}P + \sigma^2}}, \sqrt{\frac{-2\sigma^2 \ln(P_F)}{\vt{w}^H\mt{C}\vt{w}P + \sigma^2}}\right),
\end{equation}
where $Q_1(a,b)=\int_b^\infty t e^{-\frac{t^2 + a^2}{2}} I_0(at) dt$ denotes the first-order Marcum Q function for $b\geq 0$, $a>0$, and $I_0(\cdot)$ denotes the modified Bessel function of the first kind of order $0$ \cite{proakis2008digital}.
\begin{remark}\label{rm:special-case}
  The probability of detection in \eqref{eq:probability-of-detection} depends on both the mean value and the variance of the end-to-end channel.
  If the variance is zero, i.e., the channel comprises the \gls{los} path only, the second parameter of $Q_1$ in \eqref{eq:probability-of-detection} becomes independent of $\vt{w}$. Then, $P_D$ simplifies to the expression in \cite[Eq. (19)]{laue2021irsassistedactive}, where the scattered paths of the device-\gls{ris} channels were not taken into account.
\end{remark}

\section{Phase-Shift Design}\label{sec:phase-shift-design}
\noindent In this section, we first formulate a non-convex optimization problem to obtain the optimal phase-shift design for device activity detection.
In order to tackle the non-convexity, we reformulate the problem and propose two approximations of the objective function, each leading to a different suboptimal phase-shift design.
The reformulated problem is solved with an algorithm exploiting the \gls{mm} principle.
Moreover, we investigate the maximization of the average channel gain as an alternative phase-shift design criterion.

\subsection{Optimal Phase-Shift Design}
\noindent Recall that the \gls{ris} is deployed to assist the detection of devices that try to connect to the \gls{ap}, where the exact locations of the devices are not known.
Therefore, we aim for a phase-shift design that maximizes the guaranteed detection performance, i.e., the minimum probability of detection for any device location within the coverage area.
Thus, the optimal phase-shift design is obtained based on
\begin{subequations}\label{eq:optimization-problem}
  \begin{align}
    \max_{\vt{w}\in\mathbb{C}^U}\ \min_{q_t\in\mathcal{Q}}&\ Q_1\left(\sqrt{2P\frac{\abs*{\vt{w}^H\vt{h}^{\text{LoS}}}^2}{\vt{w}^H\mt{C}\vt{w}P + \sigma^2}}, \sqrt{\frac{-2\sigma^2 \ln(P_F)}{\vt{w}^H\mt{C}\vt{w}P + \sigma^2}}\right)\label{eq:objective}\\
    \text{s.t.}&\ \abs*{[\vt{w}]_u} = 1, \forall\ u\in\{0,\dots,U-1\}\label{eq:constraint-unit-modulus}.
  \end{align}
\end{subequations}

Optimization problem \eqref{eq:optimization-problem} cannot be solved directly because both constraint \eqref{eq:constraint-unit-modulus} and the Marcum Q function are non-convex.
Although the studies in \cite{kapinas2009monotonicitygeneralizedmarcum,sun2010monotonicitylogconcavity} showed that $Q_1(a, b)$ is monotonic and log-concave with respect to either $a$ or $b$, these results only hold when the other parameter is fixed.
Thus, they are not applicable to \eqref{eq:optimization-problem} because both parameters $a$ and $b$ depend on optimization variable $\vt{w}$.

Moreover, tight bounds and alternative representations of $Q_1(a, b)$ were proposed in \cite{baricz2009newboundsgeneralized} and \cite{annamalai2015newexponentialtype}, respectively, which provide simplified analytical expressions of the Marcum Q function.
However, these results do not lead to a tractable form for the optimization problem in \eqref{eq:optimization-problem}.

Therefore, we propose two approximations of the Marcum Q function that facilitate solving \eqref{eq:optimization-problem}.
We denote the resulting approximate objective functions, which will be derived in Sections \ref{sec:objective1} and \ref{sec:objective2}, respectively, as $J_k(\mt{W}, q_t),k\in\{1,2\}$, where $\mt{W}=\vt{w}\vt{w}^H\in\mathbb{C}^{U\times U}$ and $q_t\in\mathcal{Q}$ denote optimization variables.

\subsection{Reformulation to Min-Max SDP}
\noindent Let $J_k(\mt{W}, q_t)$ denote an approximation of the negative\footnote{Note the change of objective from max-min to min-max in \eqref{eq:prob-minmax-sdp}.} Marcum Q function in \eqref{eq:optimization-problem}, where $\mt{W}=\vt{w}\vt{w}^H$.
Then, \eqref{eq:optimization-problem} can be written as follows
\begin{subequations}\label{eq:prob-minmax-sdp}
  \begin{align}
    \min_{\mt{W}\in\mathbb{C}^{U\times U}}\ \max_{q_t\in\mathcal{Q}}&\ J_k(\mt{W}, q_t)\label{eq:minmax-objective}\\
    \text{s.t.}&\ \mt{W} \succeq 0\\
    &\ \left[\mt{W}\right]_{u\times u} = 1, \forall\ u\in\{0,\dots,U-1\}\label{eq:prob-minmax-sdp_diag}\\
    &\ \rank{\mt{W}} = 1,\label{eq:prob-minmax-sdp_rank}
  \end{align}
\end{subequations}
which is the well-known \gls{sdp} representation of a min-max optimization problem with a unit-modulus optimization variable. Constraint \eqref{eq:prob-minmax-sdp_rank} guarantees that the optimal vector $\vt{w}$ can be obtained from the decomposition of $\mt{W}$ and \eqref{eq:prob-minmax-sdp_diag} ensures that the elements of $\vt{w}$ have unit magnitude.
Note that problem \eqref{eq:prob-minmax-sdp} is still non-convex due to constraint \eqref{eq:prob-minmax-sdp_rank}, but we will show in Section \ref{sec:successive-convex-approximation} that this constraint can be replaced by adding a convex penalty term to the objective function.

\subsection{Lower Bound}\label{sec:objective1}
\noindent As a first approximation, we adopt a lower bound of $Q_1(a,b)$, which maintains the general goal of the desired phase-shift design, i.e., maximizing a guaranteed probability of detection.
It has been shown that \cite{annamalai2008simpleexponentialintegral,kapinas2009monotonicitygeneralizedmarcum}
\begin{equation}\label{eq:tailbound}
  Q(b+a) + Q(b-a) = Q_{0.5}(a,b) < Q_1(a,b)
\end{equation}
for $a \geq 0$ and $b>0$, where $Q(x)=\frac{1}{\sqrt{2\pi}} \int_x^\infty e^{-\frac{t^2}{2}} dt$ denotes the Gaussian Q-function.
Since $0 \leq Q(b+a)\leq Q(b-a)$ holds, we neglect $Q(b+a)$ in \eqref{eq:tailbound} and define the lower bound
\begin{equation}\label{eq:b-minus-a}
  Q(b-a) < Q_1(a,b),
\end{equation}
where
\begin{align}
  a&= \sqrt{2P\frac{\tr{\mt{W}\mt{M}}}{\tr{\mt{W}\mt{C}}P + \sigma^2}}\label{eq:a}\\
  b&= \sqrt{\frac{-2\sigma^2 \ln(P_F)}{\tr{\mt{W}\mt{C}}P + \sigma^2}}\label{eq:b}
\end{align}
are found from \eqref{eq:probability-of-detection} with $\mt{M}=\vt{h}^{\text{LoS}}\left(\vt{h}^{\text{LoS}}\right)^H$.
Note that \eqref{eq:b-minus-a} is tight\footnote{The accuracies of the bounds in \eqref{eq:tailbound} and \eqref{eq:b-minus-a} are discussed in Section \ref{sec:accuracies}.} if the \gls{los} component of the channel is dominant compared to the \gls{nlos} component, because $Q(b-a) \approx Q_1(a,b)$ for $b \gg 1$ and $b \gg b - a$ \cite{proakis2008digital}.
The probability of detection is maximized when the argument $b-a$ in \eqref{eq:b-minus-a} is minimized, leading to objective function
\begin{equation}
    \tilde{J}_1(\mt{W}, q_t) = \frac{\sqrt{-2\sigma^2 \ln(P_F)}
      - \sqrt{2P\tr{\mt{W}\mt{M}}}}
      {\sqrt{\tr{\mt{W}\mt{C}}P + \sigma^2}}.
\end{equation}
The fraction is avoided with the equivalent objective function
\begin{equation}\label{eq:objective1}
  J_1(\mt{W}, q_t) = \ln\left(\sqrt{\tr{\mt{W}\mt{C}}P + \sigma^2}\right) - \ln\left(\sqrt{2P\tr{\mt{W}\mt{M}}} - \sqrt{-2\sigma^2 \ln(P_F)}\right),
\end{equation}
which is used in \eqref{eq:prob-minmax-sdp} to obtain the optimal phase-shift vector for the lower-bound approximation of \eqref{eq:probability-of-detection}.

From \eqref{eq:objective1} we observe that minimizing $J_1(\mt{W}, q_t)$ implies minimizing the contribution of the \gls{nlos} component and maximizing the contribution of the \gls{los} component, which are reflected in the first and second term of \eqref{eq:objective1}, respectively.
In addition, we note that solving \eqref{eq:prob-minmax-sdp} with $J_1(\mt{W}, q_t)$ as objective function results in a phase-shift vector that depends on the noise power, the transmit power, the probability of false alarm, and the channel statistics.

\subsection{Objective Function for Strong LoS Path}\label{sec:objective2}
\noindent The advantage of objective function \eqref{eq:objective1} is that both the system parameters and the channel statistics are taken into account, which leads to phase-shift designs that accurately match a specific scenario.
However, in practice, phase-shift designs that are applicable for different parameter sets may be preferable.
For example, a design approach that is independent of the device's transmit power allows a more versatile deployment of the \gls{ris}.

This can be accomplished by neglecting the scattered paths of the channel, which is a reasonable approximation for channels with a strong \gls{los} component\footnote{The impact of neglecting the scattered paths is evaluated in Section \ref{sec:accuracies}.}.
Thus, by assuming $\tr{\mt{W}\mt{C}}\approx 0$, $J_1(\mt{W}, q_t)$ becomes a monotonic function in $\tr{\mt{W}\mt{M}}$, which allows us to define the equivalent objective function
\begin{equation}\label{eq:objective2}
  J_2(\mt{W}, q_t) = -\tr{\mt{W}\mt{M}}.
\end{equation}
Similarly, it is straightforward to verify that $P_D$ in \eqref{eq:probability-of-detection} is monotonic in $\abs*{\vt{w}^H\vt{h}^{\text{LoS}}}^2$ if $\vt{w}^H\mt{C}\vt{w}\approx 0$.
Objective function \eqref{eq:objective2} only depends on the \gls{los} component of the channel and is not affected by the transmit power, the noise power, and the desired probability of false alarm.
We note that $J_2(\mt{W}, q_t)$ has been used as objective function in the conference version of this paper \cite[Eq. (20)]{laue2021irsassistedactive}, cf. Remark~\ref{rm:special-case}.

\subsection{Average Channel Gain}\label{sec:channel-gain}
\noindent For \gls{ris} phase-shift design, besides the maximization of the probability of detection, the maximization of the average channel gain seems to be a desirable objective \cite{shen2021beamformingoptimizationirs,xing2021achievablerateanalysis,xing2021locationawarebeamforming}. The average channel gain is given by $\E{\abs*{h}^2} = \abs*{\vt{w}^H\vt{h}^{\text{LoS}}}^2 + \vt{w}^H\mt{C}\vt{w} = \tr{\mt{W}\left(\mt{M}+\mt{C}\right)}$, which can be maximized for all locations $q_t\in\mathcal{Q}$ using objective function\footnote{
One can find an upper bound of $P_D$ based on Markov's inequality, given by $P_D \leq (\tr{\mt{W}\left(\mt{M}+\mt{C}\right)}P + \sigma^2)/(-\sigma^2 \ln(P_F))$.
Interestingly, maximizing this bound is equivalent to minimizing $J_3(\mt{W}, q_t)$.
}
\begin{equation}\label{eq:objective3}
  J_3(\mt{W}, q_t) = -\tr{\mt{W}\left(\mt{M} + \mt{C}\right)}
\end{equation}
in \eqref{eq:prob-minmax-sdp}.
We consider $J_3(\mt{W}, q_t)$ for the performance evaluation in Section \ref{sec:performance-evaluation} and compare it to the results obtained with $J_1(\mt{W}, q_t)$ and $J_2(\mt{W}, q_t)$.
Similar to $J_2(\mt{W}, q_t)$, $J_3(\mt{W}, q_t)$ does not depend on the transmit power, noise power, and probability of false alarm.
In contrast, $J_3(\mt{W}, q_t)$ is only affected by the statistics of the channel.

\section{Optimization Algorithm}\label{sec:successive-convex-approximation}
\noindent In the previous section, three objective functions for the min-max \gls{sdp} in \eqref{eq:prob-minmax-sdp} were proposed. However, problem \eqref{eq:prob-minmax-sdp} cannot be solved directly because \eqref{eq:prob-minmax-sdp_rank} is a non-convex constraint.
Therefore, we reformulate \eqref{eq:prob-minmax-sdp} in this section and exploit the \gls{mm} principle to obtain a suboptimal phase-shift design for each objective function.

\subsection{Problem Reformulation}
\noindent Problem \eqref{eq:prob-minmax-sdp} can be equivalently rewritten as
\begin{subequations}\label{prob:min-sdp}
  \begin{align}
    \min_{\mt{W}\in\mathbb{C}^{U\times U},m\in\mathbb{R}}&\ m + \rho\left(\norm{\mt{W}}_* - \norm{\mt{W}}_2\right)\label{eq:penalty-objective}\\
    \text{s.t.}&\ \mt{W} \succeq 0\\
    &\ \left[\mt{W}\right]_{u\times u} = 1, \forall\ u\in\{0,\dots,U-1\}\\
    &\ J_k(\mt{W}, q_t) \leq m, \forall\ q_t\in\mathcal{Q}\label{eq:max-operation},
  \end{align}
\end{subequations}
where the $\max$ operation in \eqref{eq:minmax-objective} is replaced by constraint \eqref{eq:max-operation} and rank-one constraint \eqref{eq:prob-minmax-sdp_rank} is rewritten as a penalty term in \eqref{eq:penalty-objective}.
More specifically, exploiting $\norm{\mt{W}}_* - \norm{\mt{W}}_2=0 \Leftrightarrow \rank{\mt{W}} = 1$, it has been shown that the solution of \eqref{prob:min-sdp} is rank one if penalty factor $\rho>0$ is sufficiently large \cite{yu2020robustsecurewireless}.
% ensures that $\norm{\mt{W}}_* - \norm{\mt{W}}_2=0$ holds, which is equivalent to $\rank{\mt{W}} = 1$  \cite{jiang2019overaircomputation}.

In order to obtain optimized phase-shift designs based on \eqref{prob:min-sdp}, we examine the convexity of \eqref{prob:min-sdp} for $k\in\{1,2,3\}$.
We note that $J_2(\mt{W}, q_t)$ and $J_3(\mt{W}, q_t)$ are convex functions, and that $J_1(\mt{W}, q_t)$ and the objective function in \eqref{eq:penalty-objective} are in \gls{dc} form.
Thus, \eqref{prob:min-sdp} can be solved using the \gls{mm} principle, i.e., we apply convex upper bounds to the \gls{dc} terms and iteratively solve the resulting optimization problem until convergence \cite{sun2017majorizationminimizationalgorithms}.
The upper bounds are found as follows.
% However, \eqref{prob:min-sdp} cannot be solved directly because $\norm{\mt{W}}_* - \norm{\mt{W}}_2$ and $J_1(\mt{W}, q_t)$ are non-convex.
% In order to still find an optimized phase-shift design based on \eqref{prob:min-sdp}, we exploit that $\norm{\mt{W}}_* - \norm{\mt{W}}_2$ and $J_1(\mt{W}, q_t)$ are in \gls{dc} form, i.e., \eqref{prob:min-sdp} can be solved using the \gls{mm} principle \cite{sun2017majorizationminimizationalgorithms}.

\begin{lemma}\label{lm:dc}
  Let $f(\mt{W}) = f_\text{cvx}(\mt{W}) + f_\text{ccv}(\mt{W})$ denote a function in \gls{dc} form, i.e., $f_\text{cvx}(\mt{W})$ and $f_\text{ccv}(\mt{W})$ denote a convex and a concave function in $\mt{W}\succeq 0$, respectively.
  Then, $f(\mt{W})$ is bounded as
  \begin{equation}\label{eq:dcbound}
    f(\mt{W}) \leq \bar{f}(\mt{W}) = f_\text{cvx}(\mt{W}) + \bar{f}_\text{ccv}(\mt{W}),
  \end{equation}
  where $\bar{f}(\mt{W})$ is a convex function in $\mt{W}$ and
  \begin{equation}\label{eq:dc-concave-bound}
    \bar{f}_\text{ccv}(\mt{W}) = f_\text{ccv}(\mt{W}^{\prime}) + \tr{\nabla_{\mt{W}}^H f_\text{ccv}(\mt{W}^{\prime})\left(\mt{W} - \mt{W}^{\prime}\right)}
  \end{equation}
  denotes the first-order approximation of $f_\text{ccv}(\mt{W})$ at $\mt{W}^{\prime}\succeq 0$.
\end{lemma}
\begin{proof}
  Since $f_\text{ccv}(\mt{W})$ is a concave function and $\mt{W},\mt{W}^{\prime}$ are from a convex set, the first-order approximation of $f_\text{ccv}(\mt{W})$ at $\mt{W}^{\prime}$ is a global overestimator of $f_\text{ccv}(\mt{W})$, i.e., $f_\text{ccv}(\mt{W}) \leq \bar{f}_\text{ccv}(\mt{W})$ \cite{boyd2004convexoptimization}.
  Adding $f_\text{cvx}(\mt{W})$ to this inequality results in \eqref{eq:dcbound}.
\end{proof}

\begin{corollary}\label{cl:penalty}
  Given $\mt{W}^{\prime} \succeq 0$, a convex upper bound of the objective function in \eqref{eq:penalty-objective} is given by
  \begin{equation}\label{eq:dc-objective-bounded}
    m + \rho\tr*{\left(\mt{I} - \vt{v}^{\prime}{\vt{v}^{\prime}}^H\right)\mt{W}}
    - \rho\norm{\mt{W}^{\prime}}_2
    + \rho\tr*{\vt{v}^{\prime}{\vt{v}^{\prime}}^H \mt{W}^{\prime}},
  \end{equation}
  where $\vt{v}^{\prime}$ denotes the eigenvector that corresponds to the largest eigenvalue of $\mt{W}^{\prime}$.
\end{corollary}
\begin{proof}
  Using Lemma \ref{lm:dc}, we find
  \begin{equation}\label{eq:dc-bound}
    \norm{\mt{W}}_* - \norm{\mt{W}}_2 \leq
    \norm{\mt{W}}_* - \norm{\mt{W}^{\prime}}_2 - \tr{\vt{v}^{\prime}{\vt{v}^{\prime}}^H\left(\mt{W} - \mt{W}^{\prime}\right)}
  \end{equation}
  because $\nabla_{\mt{W}} \norm{\mt{W}^{\prime}}_2 = \vt{v}^{\prime}{\vt{v}^{\prime}}^H$ \cite{yang2019federatedlearningbased,xu2020resourceallocationirs}.
  Using the identity $\norm{\cdot}_* = \tr{\cdot}$ for positive semidefinite matrices and substituting the bound in \eqref{eq:dc-bound} into \eqref{eq:penalty-objective} results in \eqref{eq:dc-objective-bounded}.
\end{proof}

\begin{corollary}\label{cl:objective}
  Given $\mt{W}^{\prime} \succeq 0$, objective function $J_1(\mt{W}, q_t)$ is bounded as
  \begin{multline}\label{eq:objective1-bound}
    J_1(\mt{W}, q_t) \leq \bar{J}_1(\mt{W}, \mt{W}^{\prime}, q_t) =
      -\ln\left(
        \sqrt{2P\tr{\mt{W}\mt{M}}}
        - \sqrt{-2\sigma^2 \ln(P_F)}
        \right)\\
      + \ln\left(\sqrt{\tr{\mt{W}^{\prime}\mt{C}}P + \sigma^2}\right)
      + \frac{1}{2}\frac{\tr{\mt{C}\left(\mt{W} - \mt{W}^{\prime}\right)}P}{\tr{\mt{W}^{\prime} \mt{C}}P + \sigma^2}
  \end{multline}
\end{corollary}
\begin{proof}
  Objective function $J_1(\mt{W}, q_t)$ can be written as
  \begin{equation}\label{eq:cost-split}
    J_1(\mt{W}, q_t) = J_{1,\text{cvx}}(\mt{W}, q_t) + J_{1,\text{ccv}}(\mt{W}, q_t),
  \end{equation}
  where
  \begin{align}
    J_{1,\text{cvx}}(\mt{W}, q_t) &= -\ln\left(\sqrt{2P\tr{\mt{W}\mt{M}}} - \sqrt{-2\sigma^2 \ln(P_F)}\right)\\
    J_{1,\text{ccv}}(\mt{W}, q_t) &= \ln\left(\sqrt{\tr{\mt{W}\mt{C}}P + \sigma^2}\right)
  \end{align}
  denote a convex and a concave function, respectively.
  The gradient of $J_{1,\text{ccv}}(\mt{W}, q_t)$ is given by
  \begin{equation}
    \nabla_{\mt{W}} J_{1,\text{ccv}}(\mt{W}, q) =
      \frac{1}{2}
      \frac{\mt{C}P}{\tr{\mt{W} \mt{C}}P + \sigma^2},
  \end{equation}
  which is found by applying the chain rule while taking the derivate of $J_{1,\text{ccv}}(\mt{W}, q)$ with respect to $\mt{W}$.
  Then, \eqref{eq:objective1-bound} follows from Lemma \ref{lm:dc}.
\end{proof}

Now, we use Corollary~\ref{cl:penalty} to replace the objective function in \eqref{eq:penalty-objective} with its upper bound.
In addition, we replace $J_k(\mt{W}, q_t)$ with $\bar{J}_k(\mt{W}, \mt{W}^{\prime}, q_t)$, where $\bar{J}_1(\mt{W}, \mt{W}^{\prime}, q_t)$ is given in Corollary \ref{cl:objective}, $\bar{J}_2(\mt{W}, \mt{W}^{\prime}, q_t) = J_2(\mt{W}, q_t)$, and $\bar{J}_3(\mt{W}, \mt{W}^{\prime}, q_t) = J_3(\mt{W}, q_t)$.
The resulting convex optimization problem for $k\in\{1,2,3\}$ is given as follows
\begin{subequations}\label{prob:convex}
  \begin{align}
    \min_{\mt{W}\in\mathbb{C}^{U\times U},m\in\mathbb{R}}&\ m + \rho\tr*{\left(\mt{I} - \vt{v}^{\prime}{\vt{v}^{\prime}}^H\right) \mt{W}}\label{eq:convex-objective}\\
    \text{s.t.}&\ \mt{W} \succeq 0\\
    &\ \left[\mt{W}\right]_{u\times u} = 1, \forall\ u\in\{0,\dots,U-1\}\\
    &\ \bar{J}_k(\mt{W}, \mt{W}^{\prime}, q_t)\leq m, \forall\ q_t\in\mathcal{Q},\label{eq:cost-constraint}
  \end{align}
\end{subequations}
where \eqref{eq:convex-objective} is obtained from \eqref{eq:dc-objective-bounded} neglecting the terms not depending on $\mt{W}$ or $m$.

Then, according to the \gls{mm} principle, problem \eqref{prob:convex} is iteratively solved as follows.
Given a feasible point $\mt{W}_i=\vt{w}_i\vt{w}_i^H$ for the $i$th iteration of the algorithm, we set $\vt{v}^{\prime}=\vt{w}_i$ and $\mt{W}^{\prime}=\mt{W}_i$, and denote the solution of \eqref{prob:convex} as $\mt{W}_{i+1}$. As summarized in Algorithm~\ref{alg:sca}, these steps are repeated until the objective function converges, i.e., the relative difference between the optimal values $\Omega_i$ and $\Omega_{i-1}$ is smaller than threshold $\nu$.
\begin{algorithm}[t]
  \caption{Phase-Shift Optimization}\label{alg:sca}
  \SetKwFunction{define}{define}
  \SetKwFunction{solve}{solve}
  \SetKwFunction{decomp}{decompose}
  Define $\mathcal{Q}$, $k$, $\rho$, $\nu$, and $\vt{w}_0$\;
  $i \gets 0$\;
  $\mt{W}_i, \vt{v}_i, \Omega_i \gets \vt{w}_i \vt{w}_i^H, \vt{w}_i, \max_{q_t \in \mathcal{Q}} J_k(\mt{W}_i, q_t)$\;
  \Repeat{$\abs*{\frac{\Omega_i - \Omega_{i-1}}{\Omega_{i-1}}} \leq \nu$}{
    $\mt{W}^{\prime}, \vt{v}^{\prime} \gets \mt{W}_i, \vt{v}_i$\;
    $\mt{W}_{i+1}, m_{i+1} \gets$ \solve{problem \eqref{prob:convex}}\;
    $\vt{v}_{i+1} \gets$ \decomp{$\mt{W}_{i+1}$}\;
    $\Omega_{i+1} \gets \max_{q_t \in \mathcal{Q}} J_k(\mt{W}_{i+1}, q_t)$\;
    $i \gets i + 1$\;
  }
  $\vt{w}_\text{opt} \gets \vt{v}_i$\;
\end{algorithm}

\subsection{Convergence and Computational Complexity}
\noindent Algorithm \ref{alg:sca} follows the \gls{mm} principle and iteratively solves \eqref{prob:convex}, which tightens the upper bound of \eqref{prob:min-sdp} in each iteration.
Thus, the sequence $\{\mt{W}_i, m_i\}_{i\in\mathbb{N}}$ provides non-increasing sequences of objective values for \eqref{prob:min-sdp} and for its equivalent formulation \eqref{eq:prob-minmax-sdp}.
Furthermore, these objective values converge to a stationary value because the objective functions of \eqref{prob:min-sdp} and \eqref{eq:prob-minmax-sdp} are bounded below.
As a result, the limit point of the sequence $\{\mt{W}_i, m_i\}_{i\in\mathbb{N}}$ obtained with Algorithm \ref{alg:sca} converges to a stationary point of problem \eqref{eq:prob-minmax-sdp} \cite{sun2017majorizationminimizationalgorithms}.

The computational complexity of Algorithm \ref{alg:sca} mainly depends on the complexity of convex optimization problem \eqref{prob:convex}.
For a given solution accuracy of $\epsilon > 0$, problem \eqref{prob:convex} can be numerically solved with a worst-case complexity of $\mathcal{O}\left(\left(Q+U\right)^4\sqrt{U}\log\left(\frac{1}{\epsilon}\right)\right)$ \cite{luo2010semidefiniterelaxationquadratic}.
As one can see, the complexity polynomially depends on the number of unit cells and the number of considered device locations.
However, since \eqref{prob:convex} does not depend on instantaneous \gls{csi}, the problem does not have to be solved in an online manner.

\section{Performance Evaluation}\label{sec:performance-evaluation}
\noindent In this section, we evaluate the detection performance for different system parameters and compare the results obtained with objective functions $J_1(\mt{W} , q_t)$, $J_2(\mt{W} , q_t)$, and $J_3(\mt{W} , q_t)$.
More specifically, we apply Algorithm \ref{alg:sca} for $k\in\{1,2,3\}$, which results in three optimized phase-shift vectors $\vt{w}_{\text{opt},1}$, $\vt{w}_{\text{opt},2}$, and $\vt{w}_{\text{opt},3}$. Then, for each phase-shift vector, we use \eqref{eq:probability-of-detection} to evaluate the minimum probability of detection across all locations within the coverage area, i.e., we plot $\min_{q_t\in\mathcal{Q}} P_D(q_t)$.
For convenience, we omit the arguments $(\mt{W}, q_t)$ for the following discussion.

\subsection{Baseline Phase-Shift Design}
\noindent We also compare the performance of our proposed phase-shift designs with the quadratic phase-shift design in \cite{jamali2021powerefficiencyoverhead,laue2021irsassistedactive}, which is an analytical phase-shift design for large coverage areas.
In addition, this design is used for initialization of $\vt{w}_0$ in Algorithm~\ref{alg:sca}.
Therefore, our evaluation will reveal the performance gain achieved by the additional effort of iteratively solving optimization problem \eqref{prob:convex}.

\subsection{System and Channel Parameters}\label{sec:parameters}
\noindent The system parameters for the performance evaluation are summarized in Table~\ref{tab:parameters}.
Furthermore, since a suitable value of penalty factor $\rho$ depends on $k$, $P_\text{tx}$, $K$, etc., we select $\rho$ based on the set $\{1, 10, 100, 1000\}$ for each considered case.
\begin{table}
  \scriptsize
  \centering
  \renewcommand{\arraystretch}{1.0}
  \begin{threeparttable}
    \caption{System parameters.}
    \label{tab:parameters}
    \begin{tabular}{l l|l l|l l|l l}
      \hline
      Parameter & Value & Parameter & Value & Parameter & Value & Parameter & Value\\
      \hline
      $(c_x, c_y, c_z)$ & $(\SI{-10}{\metre}, \SI{-30}{\metre}, \SI{30}{\metre})$ &
      $(U_x, U_y)$ & $(4, 8)$ &
      $M$ & 4 &
      $\varphi_{q_t}$ & \SI{0}{\degree}\\
      $d_{q_r}, \Psi(q_r)$ & $\SI{25}{\metre}, (\SI{90}{\degree}, \SI{37}{\degree})$ &
      $(d_x, d_y)$ & $(\lambda /2, \lambda /2)$ &
      $S$ & 64 &
      $P_F$ & $0.1$\\
      $D_z$ & \SI{20}{\metre} &
      $\lambda$ & \SI{0.1}{\metre} &
      $\sigma^2$ & \SI{-100}{\dBm} & $\nu$ & $10^{-7}$\\
      \hline
    \end{tabular}
  \end{threeparttable}
\end{table}
Moreover, we study different channel conditions of the device-\gls{ris} link by varying the power ratio of the \gls{los} and \gls{nlos} components, given by $K = \abs{\bar{h}(q_t)}^2 / \sum_{l=1}^L \sigma_l^2(q_t)$.
For simplicity, we assume equal variances $\bar{\sigma}^2(q_t)=\sigma_l^2(q_t), \forall\ l\in\{1, \dots, L\}$, for all scattered paths, which results in $\bar{\sigma}^2(q_t)=\abs{\bar{h}(q_t)}^2/(KL)$.
Furthermore, we assume $L=2$ scattered paths and characterize their incident directions with angle $\alpha$ as shown in Fig.~\ref{fig:schematic}.

\begin{figure}
  \begin{minipage}[t]{0.48\columnwidth}
    \centerline{
      \tikzsetnextfilename{schematic}
      \begin{tikzpicture}
        \begin{axis}[
        	xmin=-60,
          xmax=20,
        	ymin=0,
          ymax=60,
          enlargelimits=0.01,
          font=\footnotesize,
          width=8cm,
          height=5.3cm,
          line width=1pt,
          axis line style={line width=0.5pt},
          anchor=origin,
          disabledatascaling,
          xlabel=$y$ / \si{\metre},
          ylabel=$z$ / \si{\metre},
        ]
          \pgfplotsextra{
            \newcommand{\spreading}{27}
          }

          \draw [black] (-40,20) rectangle (-20, 40);
          \draw [black] (-2, 0) -- (2, 0);

          \draw [-stealth] (135-\spreading:20) -- (135-\spreading:10);
          \draw [-stealth] (135+\spreading:20) -- (135+\spreading:10);

          \draw [dashed,line width=0.7pt] (135:70) -- (0,0);
          \draw [dashed,line width=0.7pt] (135-\spreading:70) -- (0,0);
          \draw [dashed,line width=0.7pt] (135+\spreading:70) -- (0,0);

          \draw [|-stealth,line width=0.7pt] (-18,18) arc[start angle=135,delta angle=-\spreading,x radius=25.46,y radius=25.46];
          \draw [|-stealth,line width=0.7pt] (-18,18) arc[start angle=135,delta angle=\spreading,x radius=25.46,y radius=25.46];

          \node at (-43,30) {$D_z$};
          \node at (-30,17) {$D_y$};

          \node at (135-13.5:22) {$\alpha$};
          \node at (135+13.5:22) {$-\alpha$};

          \node [pin={[pin edge={line width=0.7pt}]90:coverage area}] at (-30,35) {};
          \node [pin={[pin edge={line width=0.7pt}]181.5:scatterer directions}] at (135+\spreading:18) {};
          \node [pin={[pin edge={line width=0.7pt}]56:RIS}] at (0,0) {};

        \end{axis}
      \end{tikzpicture}
    }
    \caption{Schematic top view of the scenario.}
    \label{fig:schematic}
  \end{minipage}
  \hfill
  \begin{minipage}[t]{0.48\columnwidth}
    \centerline{
      \tikzsetnextfilename{accuracy}
      \begin{tikzpicture}
        \begin{axis}[
          line plot style,
          xlabel=$K$ / dB,
          ylabel=Relative error,
          legend entries={
            {$\epsilon(J_1, q_{t,1}^*)$},
            {$\epsilon(J_2, q_{t,2}^*)$},
            },
          legend pos=south east,
          ymin=-0.09,
          ymax=0.09,
          scaled ticks=false,
          y tick label style={/pgf/number format/fixed},
          ytick={-0.06,-0.02,0,0.02,0.06},
          extra x ticks={3},
          extra x tick style={grid=major},
          ]
        \addplot[color1,solid] table [col sep=comma] {accuracy-01_00.csv};
        \addplot[color3,dashed] table [col sep=comma] {accuracy-01_01.csv};
        \end{axis}
      \end{tikzpicture}
    }
  \caption{Relative error of approximations for $J_1$ and $J_2$ as a function of $K$.}
  \label{fig:accuracy}
  \end{minipage}
\end{figure}

\subsection{Accuracy of Approximation}\label{sec:accuracies}
\noindent First, we examine the accuracy of lower bound \eqref{eq:b-minus-a} that results in objective function $J_1$.
In addition, we quantify the effect of neglecting the scattered paths, which is used to arrive at objective function $J_2$.
To this end, let $a_k$ and $b_k$ denote parameters $a$ in \eqref{eq:a} and $b$ in \eqref{eq:b}, respectively, evaluated for $\mt{W}_k=\vt{w}_{\text{opt},k}\vt{w}_{\text{opt},k}^H$, $k\in\{1,2\}$.
Furthermore, we define $\tilde{a}_k = \left. a_k\right|_{\tr{\mt{W}_k\mt{C}}=0}$ and $\tilde{b}_k = \left. b_k\right|_{\tr{\mt{W}_k\mt{C}}=0}$, respectively, which represent $a_k$ and $b_k$ when the scattered paths are neglected.
Then, $Q(b_1 - a_1)$ is the lower bound for $Q_1(a_1, b_1)$ that leads to $J_1$ and $Q_1(\tilde{a}_2, \tilde{b}_2)$ is the approximation of $Q_1(a_2, b_2)$ that results in $J_2$.
Thus, the relative errors for $J_1$ and $J_2$ at location $q_t$ are given by $\epsilon(J_1, q_t) = Q(b_1 - a_1)/Q_1(a_1, b_1) - 1$ and $\epsilon(J_2, q_t) = Q_1(\tilde{a}_2, \tilde{b}_2)/Q_1(a_2, b_2) - 1$, respectively.
The largest relative errors across the coverage area for $J_1$ and $J_2$ are observed at $q_{t,1}^* = \argmax_{q_t\in\mathcal{Q}} \abs*{Q(b_1 - a_1) - Q_1(a_1, b_1)}/Q_1(a_1, b_1)$ and $q_{t,2}^* = \argmax_{q_t\in\mathcal{Q}} \abs*{ Q_1(\tilde{a}_2, \tilde{b}_2) - Q_1(a_2, b_2)}/Q_1(a_2, b_2)$, respectively.

For the system parameters described in Section~\ref{sec:parameters} and Table~\ref{tab:parameters}, Fig.~\ref{fig:accuracy} shows $\epsilon(J_1, q_{t,1}^*)$ and $\epsilon(J_2, q_{t,2}^*)$ for different values of $K$, where we set $D_y=\SI{20}{\metre}$, $P_\text{tx}=\SI{-1}{\dBm}$, and $\alpha=\SI{30}{\degree}$.
One can see that the relative error for $J_1$ remains below $0.02$ for all $K$.
Furthermore, we observe that the relative error for $J_2$ approaches zero for large $K$, but increases for low $K$.
However, for the performance evaluation, we consider RIS deployments that result in channels with a dominant \gls{los} component \cite{liu2020matrixcalibrationbased,mursia2021rismareconfigurableintelligent}.
Thus, assuming $K\geq\SI{3}{\decibel}$, Fig.~\ref{fig:accuracy} confirms that $J_1$ and are $J_2$ are accurate approximations with relative errors smaller than $0.06$.

\subsection{Impact of Transmit Power}
\noindent The dependence of the detection performance on the transmit power is shown in Fig. \ref{fig:var-P}.
As expected, we observe an improvement of the detection performance for all phase-shift designs as the transmit power increases. Furthermore, for high transmit powers, the minimum probability of detection across the coverage area approaches 1 in all cases.
In addition, Fig.~\ref{fig:var-P} indicates that phase shifts designed based on $J_1$ and $J_2$ outperform those designed based on $J_3$ and the baseline quadratic design.
The best performance is achieved with $J_1$, which improves the detection performance of the baseline quadratic design from 0.91 to 0.99 for $P_\text{tx}=\SI{0}{\dBm}$.
\begin{figure}
  \begin{minipage}[t]{0.48\columnwidth}
    \centerline{
      \tikzsetnextfilename{var_P}
      \begin{tikzpicture}
        \begin{axis}[
          line plot style,
          xlabel=$P_\text{tx}$ / dBm,
          ylabel=$\min_{q_t\in\mathcal{Q}} P_D(q_t)$,
          legend entries={
            {$J_1(\mt{W},q_t)$},
            {$J_2(\mt{W},q_t)$},
            {$J_3(\mt{W},q_t)$},
            {quadratic}},
          legend pos=south east,
          xmin=-4,
          xmax=4,
          ]
        \addplot table [col sep=comma] {var_P-00_00.csv};
        \addplot table [col sep=comma] {var_P-00_01.csv};
        \addplot table [col sep=comma] {var_P-00_02.csv};
        \addplot table [col sep=comma] {var_P-00_03.csv};
        \end{axis}
      \end{tikzpicture}
    }
    \caption{Minimum probability of detection as a function of the transmit power for $K=\SI{3}{\dB}$, $\alpha=\SI{30}{\degree}$, and $D_y=\SI{20}{\metre}$.}
    \label{fig:var-P}
  \end{minipage}
  \hfill
  \begin{minipage}[t]{0.48\columnwidth}
    \centerline{
      \tikzsetnextfilename{var_K}
      \begin{tikzpicture}
        \begin{axis}[
          line plot style,
          xlabel=$K$ / dB,
          ylabel=$\min_{q_t\in\mathcal{Q}} P_D(q_t)$,
          legend entries={
            {$J_1(\mt{W},q_t)$},
            {$J_2(\mt{W},q_t)$},
            {$J_3(\mt{W},q_t)$},
            {quadratic}},
          legend pos=south east,
          ymin=0.83,
          ymax=1.03
          ]
        \addplot table [col sep=comma] {var_K-00_00.csv};
        \addplot table [col sep=comma] {var_K-00_01.csv};
        \addplot table [col sep=comma] {var_K-00_02.csv};
        \addplot table [col sep=comma] {var_K-00_03.csv};

        \addplot [loosely dashdotdotted,black] coordinates {(3, 0.984) (20, 0.984)};
        \node [pin={[pin edge={line width=0.7pt}]30:$K\rightarrow\infty$}] at (6,0.984) {};
        \end{axis}
      \end{tikzpicture}
    }
    \caption{Minimum probability of detection as a function of $K$ for $D_y=\SI{20}{\metre}$, $P_\text{tx}=\SI{-1}{\dBm}$, and $\alpha=\SI{30}{\degree}$.}
    \label{fig:var-K}
  \end{minipage}
\end{figure}

The differences in performance among the optimized phase-shift designs can be explained considering the different objective functions.
More specifically, the detection performance in Fig.~\ref{fig:var-P} depends on both the transmit power and the reflection gain of the \gls{ris}.
For $J_2$ and $J_3$, the reflection gains are constant for all values of $P_\text{tx}$ because both objective functions are independent of the transmit power.
In contrast, the phase-shift design obtained for $J_1$ adapts to $P_\text{tx}$, i.e., the reflection beam can be made narrower as the transmit power increases, which yields a better performance.
Moreover, $J_2$ yields better results than $J_3$. Both objective functions maximize the power of the \gls{los} paths but $J_3$ also enhances the power of the scattered paths, which leads to more variations in the effective end-to-end channel and worse detection performance.

\subsection{Impact of Scattering Strength}\label{sec:impact-k}
\noindent The relation between the \gls{los} and the \gls{nlos} components in the channel also plays an important role for phase-shift optimization, which is illustrated in Fig.~\ref{fig:var-K}.
One can see that the detection performance increases with the value of $K$ and that all optimized phase-shift designs converge for $K\rightarrow \infty$ to the asymptotic performance of $0.984$.
This behavior can be explained with the reduced variations in the channel as $K$ increases. For $K\rightarrow \infty$, $\vt{w}^H\mt{C}\vt{w} \rightarrow 0$ and the channel becomes fully deterministic.
In this case, $J_1$, $J_2$, and $J_3$ are equivalent.

For small values of $K$, similar to Fig.~\ref{fig:var-P}, $J_1$ achieves the highest probability of detection, followed by $J_2$, $J_3$, and the baseline quadratic design.

\subsection{Impact of Scattering Directions}
\noindent In addition to the scattering strength, the incident directions of the scattered paths have an impact on the detection performance, too.
In order to evaluate this dependence in more detail, Fig. \ref{fig:var-alpha} depicts the probability of detection for $\alpha\in\{\SI{10}{\degree},\SI{20}{\degree},\SI{30}{\degree}\}$, which corresponds to scattering within the coverage area, near the edge of the coverage area, and with large distance to the coverage area, respectively, see Fig.~\ref{fig:schematic}.
Interestingly, $J_2$ provides the best performance for $\alpha = \SI{10}{\degree}$, but $J_1$ is superior for $\alpha\in\{\SI{20}{\degree},\SI{30}{\degree}\}$.
This observation can be explained as follows.

Recall that $J_1$ is minimized by maximizing the power of the \gls{los} paths and minimizing the power of the \gls{nlos} paths, which can be done concurrently and independently if both channel components are separable in the angular domain. For $\alpha\in\{\SI{20}{\degree},\SI{30}{\degree}\}$, the separation is possible because the scattering is outside the coverage area.
In this case, objective function $J_1$ results in a phase-shift design that suppresses the scattered paths, and thus achieves the best detection performance.

\begin{figure}
  % \begin{subfigure}{0.5\columnwidth}
    \centerline{
      \tikzsetnextfilename{var_alpha}
      \begin{tikzpicture}
        \begin{axis}[
          bar plot style,
          xlabel=$\alpha$ / deg,
          ylabel=$\min_{q_t\in\mathcal{Q}} P_D(q_t)$,
          legend pos=north west,
          xmin=5,
          xmax=35,
          ymin=0.85,
          ymax=1.0,
          height=5.8cm,
          ]
        \addplot[fill,color1] table [col sep=comma] {var_A-00_00.csv};
        \addplot[fill,color2] table [col sep=comma] {var_A-00_01.csv};
        \addplot[fill,color3] table [col sep=comma] {var_A-00_02.csv};
        \addplot[fill,color4] table [col sep=comma] {var_A-00_03.csv};
        \legend{
          {$J_1(\mt{W},q_t)$},
          {$J_2(\mt{W},q_t)$},
          {$J_3(\mt{W},q_t)$},
          {quadratic}}
        \end{axis}
      \end{tikzpicture}
    }
    \caption{Minimum probability of detection for $\alpha\in\{\SI{10}{\degree},\SI{20}{\degree},\SI{30}{\degree}\}$, $D_y=\SI{20}{\metre}$, $P_\text{tx}=\SI{0}{\dBm}$, and $K=\SI{3}{\decibel}$.}
    \label{fig:var-alpha}
  % \end{subfigure}
  % \begin{subfigure}{0.5\columnwidth}
  %   \centerline{
  %     \tikzsetnextfilename{accuracy}
  %     \begin{tikzpicture}
  %       \begin{axis}[
  %         bar plot style,
  %         xlabel=$k$,
  %         legend pos=south east,
  %         xmin=0.3,
  %         xmax=2.7,
  %         ymin=0.6,
  %         ymax=1.0,
  %         ]
  %       \addplot[fill,color1] table [col sep=comma] {accuracy-00_00.csv};
  %       \addplot[fill,color2] table [col sep=comma] {accuracy-00_01.csv};
  %       \addplot[fill,color3] table [col sep=comma] {accuracy-00_02.csv};
  %       \legend{
  %         {$A=\min_{q_t} P_D(\mt{W}, q_t)$},
  %         {$B=\min_{q_t} Q(b-a)$},
  %         {$C=\min_{q_t} P_D(\mt{W}, q_t), \mt{C}=\mt{0}$}}
  %       \end{axis}
  %     \end{tikzpicture}
  %   }
  % \caption{Minimum probability of detection and the approximations for $J_k$ and $\alpha=\SI{10}{\degree}$.}
  % \label{fig:accuracy}
  % \end{subfigure}
  % \caption{Results for $K=\SI{3}{\dB}$, $P_\text{tx}=\SI{0}{\dBm}$, and $D_y=\SI{20}{\metre}$.}
  % \label{fig:alpha}
\end{figure}

However, separation is not possible for $\alpha = \SI{10}{\degree}$.
Consequently, there is a design conflict for $J_1$ because minimizing the power of the scattered paths also affects the reflection gain for some \gls{los} paths.
Therefore, the reflection gain for some device locations is reduced, which limits the minimum probability of detection across the coverage area.

In contrast, $J_2$ is independent of $\alpha$, i.e., the phase-shift design considers the \gls{los} paths to the coverage area only. As a result, the \gls{los} and \gls{nlos} paths are equally reflected to the \gls{ap} when the scattering is within the coverage area.
Although this approach leads to more variations in the device-\gls{ris} channel, we observe in Fig.~\ref{fig:var-alpha} that it provides the highest minimum probability of detection for $\alpha=\SI{10}{\degree}$. However, as $\alpha$ increases, the best performance is achieved by $J_1$, which explicitly reduces the impact of the scattered paths.

\subsection{Impact of Side Lobes}
\noindent In order to highlight the difference between objective functions $J_1$ and $J_2$, Fig.~\ref{fig:reflection-patterns} illustrates the reflection pattern of the \gls{ris} for the phase shifts designed with both $J_1$ and $J_2$, respectively, i.e., it displays the reflection gain $\abs*{g\left(\Psi\left(q_r\right), \Psi\left(q_t\right)\right)}^2$ for all locations $\{q_t \mid x_{q_t} = \SI{-10}{\metre}, \SI{-60}{\metre} < y_{q_t} < \SI{20}{\metre}, \SI{0}{\metre} < z_{q_t} < \SI{60}{\metre}\}$.
As in the schematic view in Fig. \ref{fig:schematic}, the arrows in Fig.~\ref{fig:reflection-patterns} indicate the incident directions of the two scattered paths and the rectangle marks the coverage area.
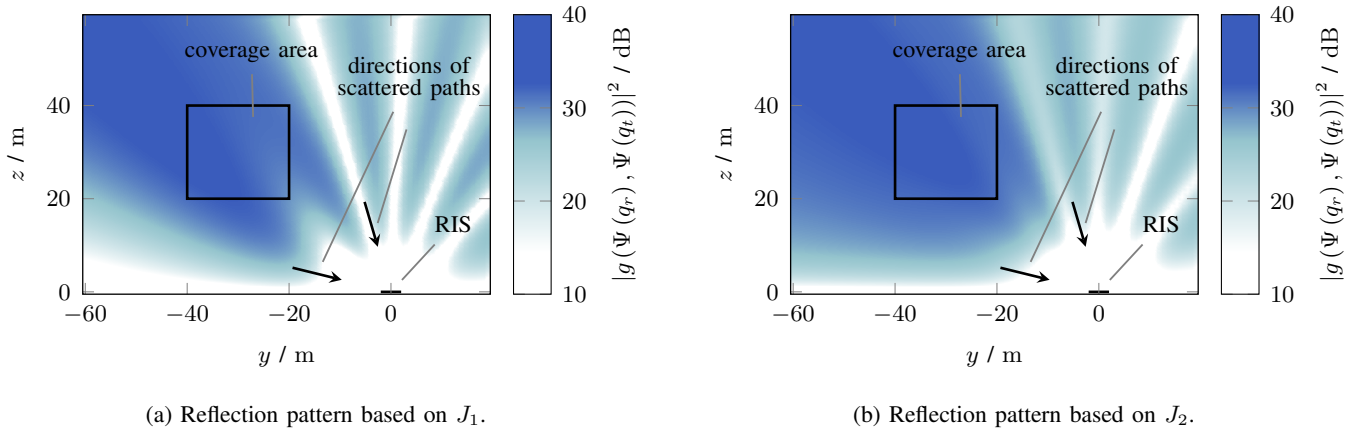
\begin{figure}
  \begin{subfigure}{0.43\columnwidth}
    \centerline{
      \tikzsetnextfilename{reflection_pattern_j1}
      \begin{tikzpicture}
        % Scatterer for alpha = 30.0 deg
        \pgfmathparse{-20*sin(45-30.0)}
        \let\OneAxpos=\pgfmathresult
        \pgfmathparse{20*cos(45-30.0)}
        \let\OneAypos=\pgfmathresult

        \pgfmathparse{-10*sin(45-30.0)}
        \let\OneBxpos=\pgfmathresult
        \pgfmathparse{10*cos(45-30.0)}
        \let\OneBypos=\pgfmathresult

        \pgfmathparse{-20*sin(45+30.0)}
        \let\TwoAxpos=\pgfmathresult
        \pgfmathparse{20*cos(45+30.0)}
        \let\TwoAypos=\pgfmathresult

        \pgfmathparse{-10*sin(45+30.0)}
        \let\TwoBxpos=\pgfmathresult
        \pgfmathparse{10*cos(45+30.0)}
        \let\TwoBypos=\pgfmathresult

        \pgfmathparse{400*pi}
        \let\gscale=\pgfmathresult

        \begin{axis}[heat map style,colorbar style={ylabel=$\abs*{g\left(\Psi\left(q_r\right), \Psi\left(q_t\right)\right)}^2$ / \si{\decibel},font=\footnotesize,}]
        \addplot[matrix plot*,shader=interp,mesh/cols=80,point meta=explicit] table [col sep=comma,meta=z,point meta min=10,point meta max=40] {pattern-00.csv};
        \draw [black] (-40,20) rectangle (-20, 40);
        \draw [black,-stealth] (\OneAxpos, \OneAypos) -- (\OneBxpos, \OneBypos);
        \draw [black,-stealth] (\TwoAxpos, \TwoAypos) -- (\TwoBxpos, \TwoBypos);
        \draw [black] (-2, 0) -- (2, 0);

        \pgfplotsextra{
          \newcommand{\spreading}{30}
        }
        \node [pin={[pin edge={line width=0.7pt}]92:coverage area}] at (-27,35) {};
        \node [pin={[pin edge={line width=0.7pt},align=center,pin distance=2cm,font=\footnotesize\linespread{0.8}\selectfont]87:{directions of\\scattered paths}}] at (135+\spreading:15) {};
        \draw[line width=0.7pt,color=gray] (135-5-\spreading:15) -- (85:35);
        \node [pin={[pin edge={line width=0.7pt}]56:RIS}] at (0,0) {};

        \end{axis}
      \end{tikzpicture}
    }
    \caption{Reflection pattern based on $J_1$.}
    \label{fig:reflection-pattern-j1}
  \end{subfigure}
  \hfill
  \begin{subfigure}{0.43\columnwidth}
    \centerline{
      \tikzsetnextfilename{reflection_pattern_j2}
      \begin{tikzpicture}
        % Scatterer for alpha = 30.0 deg
        \pgfmathparse{-20*sin(45-30.0)}
        \let\OneAxpos=\pgfmathresult
        \pgfmathparse{20*cos(45-30.0)}
        \let\OneAypos=\pgfmathresult

        \pgfmathparse{-10*sin(45-30.0)}
        \let\OneBxpos=\pgfmathresult
        \pgfmathparse{10*cos(45-30.0)}
        \let\OneBypos=\pgfmathresult

        \pgfmathparse{-20*sin(45+30.0)}
        \let\TwoAxpos=\pgfmathresult
        \pgfmathparse{20*cos(45+30.0)}
        \let\TwoAypos=\pgfmathresult

        \pgfmathparse{-10*sin(45+30.0)}
        \let\TwoBxpos=\pgfmathresult
        \pgfmathparse{10*cos(45+30.0)}
        \let\TwoBypos=\pgfmathresult

        \pgfmathparse{400*pi}
        \let\gscale=\pgfmathresult

        \begin{axis}[heat map style,colorbar style={ylabel=$\abs*{g\left(\Psi\left(q_r\right), \Psi\left(q_t\right)\right)}^2$ / \si{\decibel},font=\footnotesize,}]
        \addplot[matrix plot*,shader=interp,mesh/cols=80,point meta=explicit] table [col sep=comma,meta=z,point meta min=10,point meta max=40] {pattern-01.csv};

        \draw [black] (-40,20) rectangle (-20, 40);
        \draw [black,-stealth] (\OneAxpos, \OneAypos) -- (\OneBxpos, \OneBypos);
        \draw [black,-stealth] (\TwoAxpos, \TwoAypos) -- (\TwoBxpos, \TwoBypos);
        \draw [black] (-2, 0) -- (2, 0);

        \pgfplotsextra{
          \newcommand{\spreading}{30}
        }
        \node [pin={[pin edge={line width=0.7pt}]92:coverage area}] at (-27,35) {};
        \node [pin={[pin edge={line width=0.7pt},align=center,pin distance=2cm,font=\footnotesize\linespread{0.8}\selectfont]87:{directions of\\scattered paths}}] at (135+\spreading:15) {};
        \draw[line width=0.7pt,color=gray] (135-5-\spreading:15) -- (85:35);
        \node [pin={[pin edge={line width=0.7pt}]56:RIS}] at (0,0) {};

        \end{axis}
      \end{tikzpicture}
    }
    \caption{Reflection pattern based on $J_2$.}
    \label{fig:reflection-pattern-j2}
  \end{subfigure}
  \caption{Reflection patterns for $D_y=\SI{20}{\metre}$, $P_\text{tx}=\SI{0}{\dBm}$, $\alpha=\SI{30}{\degree}$, and $K=\SI{3}{\dB}$.}
  \label{fig:reflection-patterns}
\end{figure}

One can see in Fig.~\ref{fig:reflection-pattern-j2} that the reflection pattern for $J_2$ has a wide main lobe that is centered at the coverage area and slowly decays at both sides.
In addition, both scattered paths are covered by the side lobes of the reflection pattern, i.e., the channel variations of the device-\gls{ris} channel are reflected to the \gls{ap}. The phase-shift design used in Fig.~\ref{fig:reflection-pattern-j2} results in $\min_{q_t\in\mathcal{Q}} P_D(q_t)=0.95$.

In contrast, the reflection pattern for $J_1$ comprises two narrow main lobes that are directed to the corners of the coverage area.
Furthermore, the side lobes affecting the scattered paths in Fig.~\ref{fig:reflection-pattern-j1} are significantly reduced compared to those in Fig.~\ref{fig:reflection-pattern-j2}.
Thus, less variations of the device-\gls{ris} channel are reflected to the \gls{ap} and the detection performance increases to $\min_{q_t\in\mathcal{Q}} P_D(q_t)=0.99$.

\subsection{Impact of Area Size}
\noindent As one can observe in Fig.~\ref{fig:area-impact}, the size of the coverage area is also a limiting factor for the minimum guaranteed probability of detection.
As expected, the figure shows that the detection performance decreases with increasing $D_y$ because the reflection gain of the \gls{ris} is distributed over a larger area.
Moreover, Fig.~\ref{fig:area-impact} emphasizes the large gain that the optimized phase-shift designs achieve compared to the baseline quadratic design.

\begin{figure}
  \begin{minipage}[t]{0.48\columnwidth}
    \centerline{
      \tikzsetnextfilename{var_D}
      \begin{tikzpicture}
        \begin{axis}[
          bar plot style,
          xlabel=$D_y$ / \si{\metre},
          ylabel=$\min_{q_t\in\mathcal{Q}} P_D(q_t)$,
          xmin=5,
          xmax=35,
          height=5.8cm,
          ]
        \addplot[fill,color1] table [col sep=comma] {var_D-00_00.csv};
        \addplot[fill,color2] table [col sep=comma] {var_D-00_01.csv};
        \addplot[fill,color3] table [col sep=comma] {var_D-00_02.csv};
        \addplot[fill,color4] table [col sep=comma] {var_D-00_03.csv};
        \legend{
          {$J_1(\mt{W},q_t)$},
          {$J_2(\mt{W},q_t)$},
          {$J_3(\mt{W},q_t)$},
          {quadratic}}
        \end{axis}
      \end{tikzpicture}
    }
    \caption{Minimum probability of detection as a function of area width for $K=\SI{3}{\dB}$, $P_\text{tx}=\SI{-3}{\dBm}$, and $\alpha=\SI{30}{\degree}$.}
    \label{fig:area-impact}
  \end{minipage}
  \hfill
  \begin{minipage}[t]{0.48\columnwidth}
    \centerline{
      \tikzsetnextfilename{var_U}
      \begin{tikzpicture}
        \begin{axis}[
          bar plot style,
          xlabel=$U$,
          ylabel=$\min_{q_t\in\mathcal{Q}} P_D(q_t)$,
          legend pos=north west,
          xmin=24,
          xmax=72,
          ymin=0.7,
          ymax=1.05,
          height=5.8cm,
          ]
        \addplot[fill,color1] table [col sep=comma] {var_U-00_00.csv};
        \addplot[fill,color2] table [col sep=comma] {var_U-00_01.csv};
        \addplot[fill,color3] table [col sep=comma] {var_U-00_02.csv};
        \addplot[fill,color4] table [col sep=comma] {var_U-00_03.csv};
        \legend{
          {$J_1(\mt{W},q_t)$},
          {$J_2(\mt{W},q_t)$},
          {$J_3(\mt{W},q_t)$},
          {quadratic}}
        \end{axis}
      \end{tikzpicture}
    }
    \caption{Minimum probability of detection as a function of \gls{ris} cells for $K=\SI{3}{\dB}$, $P_\text{tx}=\SI{-3}{\dBm}$, $D_y=\SI{20}{\metre}$, and $\alpha=\SI{30}{\degree}$.}
    \label{fig:irs-impact}
  \end{minipage}
\end{figure}

It is worth noting that phase shifts designed based on $J_1$ and $J_2$, respectively, provide almost same the performance for $D_y=\SI{30}{\metre}$.
In this case, the scatterers are close to the corner of the coverage area, i.e., the impact of the \gls{nlos} paths cannot be minimized by $J_1$ without reducing the reflection gain of the \gls{los} paths.

\subsection{Impact of RIS Size}
\noindent Finally, we show the impact of the \gls{ris} size on the detection performance in Fig.~\ref{fig:irs-impact}.
More specifically, we show the minimum probability of detection for different numbers of unit cells $U\in\{32, 48, 64\}$.
Fig.~\ref{fig:irs-impact} indicates that the detection performance improves when the size of the \gls{ris} increases, which is expected because more signal energy can be reflected to the \gls{ap} with larger \glspl{ris}.
As $U$ changes from $32$ to $64$, we observe an improvement of at least \SI{14}{\percent} for the optimized phase-shift designs.
Moreover, the differences between the design approaches remain the same, i.e., the best performance is achieved with $J_1$, followed by $J_2$, $J_3$, and the baseline quadratic design.
The latter does not yield a large gain as $U$ increases because the phase shifts are not specifically designed to maximize the minimum probability of detection.
Consequently, the illumination of some parts of the coverage area may not be improved with larger \gls{ris} size, which limits the minimum detection performance across the coverage area.

\section{Conclusion}\label{sec:conclusion}
\noindent This paper studied device activity detection for \gls{gf} uplink transmission in \gls{ris} assisted communication systems, where the \gls{ris} was deployed to cover a specific area.
In order to optimally design the phase shifts of the \gls{ris}, we employed \gls{glrt} for device activity detection and showed that the resulting probability of detection can be expressed in terms of the Marcum Q function.
Furthermore, we formulated an optimization problem for the \gls{ris} phase shifts for maximization of the minimum probability of detection across the coverage area.
The non-convexity of the optimization problem was tackled by applying the \gls{mm} principle and using two approximations of the Marcum Q function.
Based on a lower bound of the Marcum Q function, the first approximation resulted in the best detection performance in most cases because all system parameters and the channel statistics are taken into account for phase-shift optimization.
For the second approximation, we neglected the scattered paths of the channel for phase-shift optimization and showed that a similar performance as for the first approximation is achieved when the \gls{los} paths of the channel are dominant or when the \gls{nlos} paths and \gls{los} paths of the channel share the same incident angles. In addition, the second approximation resulted in an versatile objective function for phase-shift optimization because it does not depend on system parameters such as transmit power, noise power, and probability of false alarm.
Finally, our performance evaluation revealed that maximizing the average channel gain is not a suitable phase-shift design criterion for \gls{ris} assisted device activity detection.
Although the resulting phase-shift design outperforms the baseline quadratic design, the proposed designs based on approximations of the Marcum Q function yield a significantly higher performance.

% \section*{Acknowledgment}
%
% The preferred spelling of the word ``acknowledgment'' in America is without
% an ``e'' after the ``g''. Avoid the stilted expression ``one of us (R. B.
% G.) thanks $\ldots$''. Instead, try ``R. B. G. thanks$\ldots$''. Put sponsor
% acknowledgments in the unnumbered footnote on the first page.

\bibliographystyle{IEEEtran}
\bibliography{IEEEabrv,references}

\end{document}